\documentclass[journal,draftcls,onecolumn,12pt,twoside]{IEEEtran}
\IEEEoverridecommandlockouts
\usepackage{lineno}
\usepackage{hyperref}
\usepackage{cite}
\usepackage{amsmath,amssymb,amsfonts,cases} 
\usepackage{amsmath}
\usepackage{amsthm}
\DeclareMathOperator*{\argmax}{argmax}

\usepackage{graphicx}
\usepackage{textcomp}
\usepackage{xcolor}
\usepackage{graphicx}
\usepackage{float}
\usepackage{subfigure}
\usepackage{amsmath,mleftright}
\usepackage{amsfonts,amssymb}
\usepackage{mathrsfs}
\usepackage{mathtools}
\usepackage{algorithm}
\usepackage{algorithmic}
\usepackage{bm}
\usepackage{multirow}
\usepackage{array}
\usepackage{amssymb}
\usepackage{amsmath}
\usepackage{cite}
\usepackage{url}
\usepackage{xcolor}
\usepackage{cite,graphicx,amsmath,amssymb}
\usepackage{subfigure}
\usepackage{fancyhdr}
\usepackage{mdwmath}
\usepackage{mdwtab}
\usepackage{caption}
\usepackage{amsthm}
\usepackage{setspace}
\usepackage{bm}
\usepackage{mathtools}
\usepackage{dsfont}
\usepackage{bbm}
\usepackage{framed}
\newtheorem{remark}{Remark}
\newtheorem{theorem}{Theorem}

\newtheorem{lemma}{Lemma}

\newtheorem{corollary}{Corollary}

\makeatletter
\newcommand{\biggg}{\bBigg@{3}}
\newcommand{\Biggg}{\bBigg@{3.5}}
\makeatother
\makeatletter
\renewcommand{\maketag@@@}[1]{\hbox{\m@th\normalsize\normalfont#1}}%
\makeatother
\def\BibTeX{{\rm B\kern-.05em{\sc i\kern-.025em b}\kern-.08em
    T\kern-.1667em\lower.7ex\hbox{E}\kern-.125emX}}
    \expandafter\def\expandafter\normalsize\expandafter{%
    \normalsize%
    \setlength\abovedisplayskip{4pt}%
    \setlength\belowdisplayskip{4pt}%
    \setlength\abovedisplayshortskip{2pt}%
    \setlength\belowdisplayshortskip{2pt}%
}
\allowdisplaybreaks[4]
\begin{document}
\title{Frequency-Selective Pinching-Antenna Systems}
\author{Chongjun~Ouyang, Hao~Jiang, Zhiguo~Ding, and Robert~Schober\vspace{-10pt}
\thanks{C. Ouyang and H. Jiang are with the School of Electronic Engineering and Computer Science, Queen Mary University of London, London, E1 4NS, U.K. (e-mail: \{c.ouyang, hao.jiang\}@qmul.ac.uk).}
\thanks{Z. Ding is with the School of Electrical and Electronic Engineering (EEE), Nanyang Technological University, Singapore 639798 (e-mail: zhiguo.ding@ntu.edu.sg).}
\thanks{R. Schober is with the Institute for Digital Communications, Friedrich-Alexander-Universit\"{a}t Erlangen-N\"{u}rnberg, 91058, Erlangen, Germany (e-mail: robert.schober@fau.de).}}
\maketitle
\begin{abstract}
An uplink pinching-antenna system (PASS) is considered, where a segmented dielectric waveguide is employed to serve multiple users and one pinching antenna (PA) is activated on each segment. Continuous- and discrete-time signal models are developed for pulse-shaped single-carrier transmission. Extensive waveguide deployment reduces user-to-PA path loss but also creates differential propagation delays among the PA-assisted paths spanning the large waveguide aperture. These delays are retained, and a Toeplitz block model is derived to characterize the resulting frequency-selective (FS) channel. For infinite-blocklength transmission, the achievable single-user rate and the multiuser sum-rates under joint and parallel decoding are characterized. For carrier-phase-aligned single-user PA placement, the frequency-flat (FF) model is shown to overestimate the exact FS rate. The difference between the rates computed under the FF channel model and the exact FS channel model is referred to as the \emph{FF rate overestimation}, and its dependence on the waveguide-aperture-induced delay spread is characterized. For small delay spreads, the FF rate overestimation is shown to increase approximately quadratically with the delay spread, whereas a derived lower bound reveals that the overestimation grows at least double logarithmically with the waveguide aperture, or equivalently with the delay spread, in the large-aperture regime. For finite-blocklength transmission, achievable rates are derived under a practical cyclic-prefix single-carrier frequency-domain equalization framework. An element-wise algorithm is developed to optimize FS-aware PA placement in both transmission regimes. The analytical findings are verified numerically. It is demonstrated that differential delays across the extensively deployed waveguide can cause significant frequency selectivity and that the proposed FS-aware PA placement method substantially improves the achievable rate over placement based on the FF approximation, especially for large bandwidths and apertures.
\end{abstract}
\section{Introduction}
Reconfigurable antennas have recently attracted substantial interest because they provide new degrees of freedom for wireless channel control \cite{heath2025tri}. Representative technologies include movable and fluid antennas, which alter antenna positions \cite{zhu2024movable,wong2020fluid}, as well as dynamic metasurface antennas and reconfigurable intelligent surfaces, which adjust electromagnetic (EM) responses \cite{shlezinger2021dynamic,wu2025intelligent}. These architectures can exploit spatial diversity and shape radiation patterns with modest radio-frequency (RF) complexity \cite{heath2025tri}. However, the reconfigurability of these antennas is typically confined to a wavelength-scale region. Such local reconfiguration cannot directly overcome the large-scale path loss and blockage that constrain high-frequency transmission, particularly in the upper mid-band envisioned for sixth-generation networks.

The pinching-antenna system (PASS), introduced by NTT DOCOMO, provides a different form of spatial reconfigurability \cite{suzuki2022pinching}. A dielectric waveguide carries the RF signal, and small dielectric particles attached to the waveguide act as pinching antennas (PAs). Each PA couples energy between the guided wave and free space. Its location can be adjusted along the waveguide in a manner analogous to moving a clothespin along a clothesline \cite{liu2025pinching,liu2025pinchingtutorial}. Since a dielectric waveguide may extend over tens or even hundreds of meters, a PA can be activated close to its intended user. This architecture establishes a short and stable line-of-sight (LoS) link and can substantially reduce both path loss and blockage. PASS also supports scalable antenna deployment because PAs can be added, removed, or repositioned without requiring a conventional RF chain at every radiating element.

These properties have motivated extensive research on PASS; recent surveys provide comprehensive overviews \cite{liu2025pinchingsurvey,lu2026survey}. Beyond the original dielectric-waveguide implementation, generalized PASS extends both the transmission medium and the signal-distribution architecture \cite{xu2026generalized}. Examples include leaky-wave coaxial cables \cite{xu2026generalized}, radio stripes \cite{xu2026generalized}, radio-over-fiber links \cite{beas2013millimeter}, and wireless-fed structures \cite{wijewardhana2026wireless}. Another example is the segmented waveguide-enabled pinching-antenna system (SWAN) \cite{ouyang2025uplink}, which replaces one long waveguide with multiple short segments connected to the base station (BS) through calibrated wired links. In the uplink, each segment activates one PA. A signal collected by this PA is not re-radiated through PAs on other segments, so the segmented structure eliminates the inter-antenna radiation (IAR) effect that complicates multi-PA reception on monolithic waveguides. 

Despite the different media and architectures of these generalized PASS implementations, they share the same essential principle: \emph{an extensively deployed transmission medium brings a radiating or receiving point close to the user, which reduces path loss and extends coverage}. This extensive waveguide deployment also creates an important signal-processing challenge. The PA-assisted paths spanning the large waveguide aperture traverse different free-space and in-waveguide distances and therefore arrive with different delays. Their delay spread can become comparable to the symbol interval even when the fractional bandwidth is small. A PASS channel may thus be \emph{frequency selective} (FS) rather than \emph{frequency flat} (FF). This effect is known from electrically large arrays \cite{wang2018spatial}, but PASS combines free-space propagation, guided propagation, and position-dependent coupling. 

Most PASS analyses nevertheless optimize the PA positions under an FF approximation; see \cite{liu2025pinchingsurvey,lu2026survey}. A recent study considered single-user multi-carrier PASS and emphasized frequency selectivity caused by wideband signaling and waveguide dispersion \cite{xiao2025frequency}. However, an in-depth analysis of the frequency selectivity caused by differential delays across the large waveguide aperture formed through extensive deployment remains unavailable. Moreover, a unified treatment of pulse-shaped transmission is needed to quantify the \emph{FF rate overestimation}, i.e., the difference between the rates computed under the FF and exact FS channel models, and to optimize PA placement accordingly. This paper develops such a treatment for an uplink FS PASS based on a segmented waveguide and shows that aperture-induced delays can invalidate the FF model even for moderate bandwidths. Single-carrier signaling is adopted, while higher-order waveguide dispersion is excluded, to isolate the frequency selectivity caused by aperture-induced differential propagation delays from multi-carrier effects and waveguide dispersion. The main contributions are summarized as follows.
\begin{itemize}
\item We consider an SWAN-based multiuser uplink channel and establish continuous- and discrete-time models for pulse-shaped single-carrier PASS transmission. This setting isolates the frequency selectivity created by differential path delays across the large waveguide aperture formed through extensive PA deployment. The proposed models retain these delays and the memory introduced by the bandlimited pulse. They yield a Toeplitz block model that reduces to the conventional FF model only when the aperture-induced delay spread is negligible.
\item We characterize the infinite-blocklength achievable rates for both single-user and multiuser transmission. For one user, Toeplitz asymptotics and Szeg{\H{o}}'s theorem convert the large-block mutual information (MI) into a frequency-domain integral. We then propose an element-wise algorithm to optimize the PA positions under the exact FS model. For carrier-phase-aligned PA placement, we prove that the FF model overestimates the exact FS rate. When the aperture-induced delay spread is small, we show that the FF rate overestimation increases approximately quadratically with the delay spread. We further derive a lower bound on the FF rate overestimation and use it to prove that the exact overestimation grows at least double logarithmically with the waveguide aperture, or equivalently with the delay spread, in the large-aperture regime. For multiple users, we derive achievable sum-rate expressions under joint and parallel decoding and design FS-aware PA placement separately for each decoding strategy to maximize the corresponding sum-rate.  
\item We characterize finite-block cyclic-prefix single-carrier frequency-domain equalization (CP-SC-FDE) transmission over the FS PASS channel. The required cyclic prefix accounts for both the physical path-delay spread and the effective duration of the bandlimited pulse. We derive closed-form rate expressions for single-user reception and multiuser parallel decoding. These expressions remain explicit functions of the PA positions and therefore support FS-aware PA placement optimization. Their large-block limits also reveal the rate loss caused by practical SC-FDE processing.
\item We present numerical results to verify the analytical scaling laws and demonstrate that: \romannumeral1) differential delay spread across an extensively deployed waveguide can produce appreciable frequency selectivity in practical PASS channels, which establishes frequency selectivity as a fundamental modeling and design consideration; \romannumeral2) FS-aware PA placement consistently outperforms placement based on the FF approximation, with a gain that increases with the bandwidth or aperture; and \romannumeral3) this gain is especially significant for multiuser parallel decoding.
\end{itemize}

The remainder of this paper is organized as follows. Section \ref{Section: System Model} introduces the FS PASS model. Section \ref{Section:Infinite-Blocklength Analysis} studies infinite-blocklength transmission, and Section \ref{Section: Finite-Block CP-SC-FDE} develops the finite-block CP-SC-FDE framework. Section \ref{Section:Numerical Results} presents numerical results, and Section \ref{Section: Conclusion} concludes the paper.

\noindent\emph{Notation:} Scalars, vectors, and matrices are denoted by nonbold, bold lowercase, and bold uppercase letters, respectively. For a matrix $\mathbf A$, $\mathbf A^{\mathsf T}$, $\mathbf A^*$, $\mathbf A^{\mathsf H}$, $\mathbf A^{-1}$, and $[{\mathbf{A}}]_{i,j}$ denote its transpose, conjugate, conjugate transpose, inverse, and $(i,j)$th element, respectively. Moreover, $\mathsf{diag}(\mathbf a)$ forms a diagonal matrix from $\mathbf a$, while $\mathbf F_N$, $\mathbf I_N$, and $\mathbf 0$ denote the $N\times N$ discrete Fourier transform (DFT) matrix, identity matrix, and zero matrix, respectively. The sets $\mathbb{C}$, $\mathbb{R}$, and $\mathbb{Z}$ denote the complex numbers, real numbers, and integers, respectively. The symbols $\delta[\cdot]$ and $\delta(\cdot)$ denote the Kronecker and Dirac delta functions, respectively. The operators $\mathcal F(\cdot)$, $\mathcal F^{-1}(\cdot)$, and $\ast$ denote the Fourier transform, inverse Fourier transform, and convolution, respectively. Finally, $[N]\triangleq\{1,\ldots,N\}$, and $\mathcal{CN}(\bm\mu,\mathbf X)$ denotes the circularly symmetric complex Gaussian distribution with mean $\bm\mu$ and covariance $\mathbf X$.

\begin{figure}[!t]
\centering
    \subfigure[System setup.]
    {
        \includegraphics[width=0.45\textwidth]{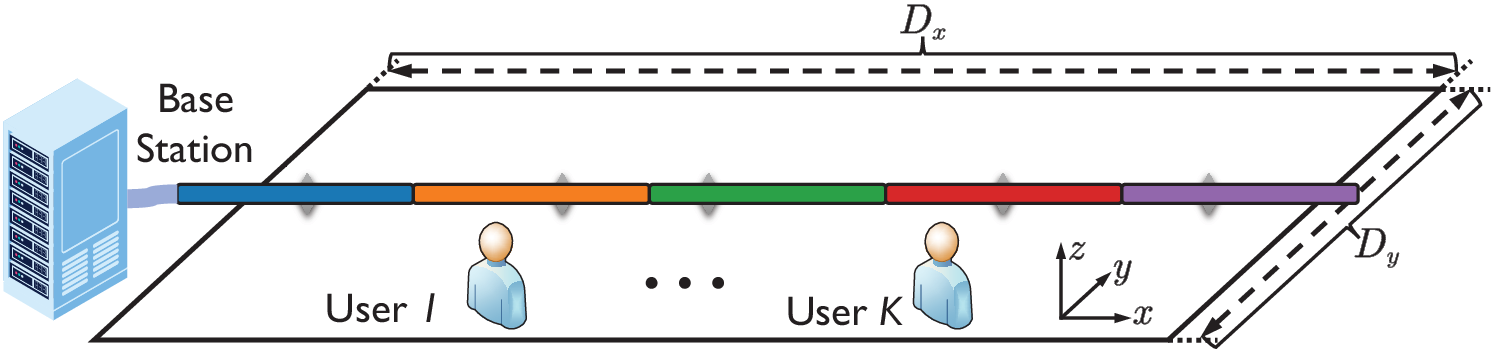}
        \label{Figure: SWAN_System_Model}
    }
    \subfigure[Segmented waveguide.]
    {
        \includegraphics[width=0.45\textwidth]{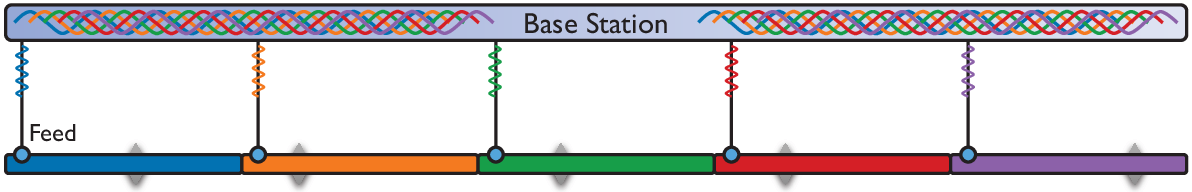}
        \label{Figure: SWAN_System_Model1}
    }
\caption{Illustration of the segmented-waveguide system \cite{ouyang2025uplink}.}
\label{Figure: New_SWAN_System_Model}
\vspace{-10pt}
\end{figure}

\section{System Model}\label{Section: System Model}
Consider an uplink PASS in which a BS employs a segmented waveguide to serve $K$ single-antenna users in a rectangular region of size $D_x\times D_y$ in the $xy$-plane. The position of user $k$ is ${\mathbf{u}}_k=[u_k^x,u_k^y,0]^{\mathsf T}$ for $k\in[K]$. The waveguide is deployed along the $x$-axis at height $d$ and provides horizontal coverage. PAs placed along the waveguide extract the uplink signals, as shown in {\figurename} {\ref{Figure: SWAN_System_Model}}.

The waveguide comprises $M$ dielectric segments of length $L$, such that $D_x=ML$. Let ${\bm\psi}_0^m\triangleq[\psi_0^m,0,d]^{\mathsf T}$ denote the feed point of segment $m$ for $m\in[M]$, where $\psi_0^1<\psi_0^2<\cdots<\psi_0^M$. For simplicity, each feed point lies at the front-left end of its segment. To focus on inter-segment frequency selectivity and eliminate the IAR effect, only one PA is active on each segment \cite{ouyang2025uplink}. The PA position on segment $m$ is ${\bm\psi}_m\triangleq[\psi_m,0,d]^{\mathsf T}$ and satisfies
\begin{align}\nonumber
&\psi_{0}^{m}\leq \psi_{m}\leq \psi_{0}^{m}+L,\quad m\in[M],\\
&\lvert\psi_{m}-\psi_{m'}\rvert\geq \Delta,\quad m,m'\in[M],m\ne m',\nonumber
\end{align}
where $\Delta>0$ is the minimum inter-antenna spacing used to mitigate mutual coupling. Each PA extracts the incident uplink signal through EM coupling \cite{wang2025modeling}. The extracted signal propagates through the corresponding dielectric segment to its feed point and is then conveyed from the feed point to the BS through a wired transport link; see {\figurename} {\ref{Figure: SWAN_System_Model1}}. An analog radio-over-fiber link provides one practical implementation \cite{beas2013millimeter}. Although these wired links can be physically long and therefore have non-negligible absolute propagation delays, their end-to-end delays can be calibrated and matched across the segments. We therefore assume that all feed-point-to-BS links have the same end-to-end propagation delay. This common delay is absorbed into the receiver timing reference and does not contribute to the differential delay spread considered in this paper. Moreover, the distance-dependent attenuation of the wired links is assumed to be much smaller than the in-waveguide attenuation over the considered deployment range and is therefore neglected \cite{beas2013millimeter}. Fixed conversion and connector losses are assumed to be calibrated and absorbed into the overall channel-gain coefficient.
\subsection{Signal Model}
We consider pulse-shaped single-carrier transmission. For a blocklength $N$, the complex baseband waveform transmitted by user $k\in[K]$ is given by
\begin{align}\label{Finite_Waveform}
x_{k,N}(t)=\sum_{n=1}^{N}s_k[n]p(t-nT_{\rm s}),
\end{align}
where $s_k[n]\in{\mathbb{C}}$ is the $n$th information-bearing symbol of user $k$, $p(t)$ is the transmit pulse with unit energy, and $T_{\rm s}$ is the symbol interval. To focus on the aperture-induced frequency selectivity while retaining analytical tractability, we adopt minimum-bandwidth Nyquist signaling with $p(t)={\sqrt{B}}\operatorname{sinc}(Bt)$, where $B=\frac{1}{T_{\rm{s}}}$ is the signal bandwidth and $\operatorname{sinc}(x)\triangleq\frac{\sin(\pi x)}{\pi x}$. Accordingly, the Fourier transform of the transmit pulse is $P(f)\triangleq{\mathcal F}(p(t))=\frac{1}{\sqrt{B}}\operatorname{rect}(\frac{f}{B})$, where $\operatorname{rect}(\cdot)$ denotes the rectangular function. The corresponding pulse autocorrelation is $c_p(t)\triangleq\int_{-\infty}^{\infty}p(u)p^*(u-t){\rm d}u=\operatorname{sinc}(Bt)$, which satisfies $c_p(qT_{\rm s})=\delta[q]$ for $q\in\mathbb{Z}$. Hence, the symbol rate $\frac{1}{T_{\rm{s}}}=B$ equals the Nyquist sampling rate of the adopted complex-baseband waveform.\footnote{For a square-root-raised-cosine (SRRC) pulse with roll-off factor $\beta>0$, the occupied bandwidth satisfies $B=\frac{1+\beta}{T_{\rm s}}$. In this case, symbol-rate sampling of the transmit-pulse matched-filter output is generally not sufficient for an arbitrary FS channel. An information-preserving representation can instead be obtained through a fractionally spaced or equivalent polyphase receiver \cite{gallager2008principles}. The resulting sufficient-statistic model can be represented using Toeplitz or block-Toeplitz structures, with a matrix-valued generating symbol in the polyphase formulation \cite{gazzah2001asymptotic}. The aperture-induced differential-delay mechanism remains unchanged. For a fixed roll-off factor, the small-delay-spread and large-aperture scaling behaviors of the FF rate overestimation are also preserved, although the corresponding coefficients and bounds become pulse dependent. We focus on the minimum-bandwidth sinc case to retain a concise scalar formulation. A detailed treatment of the SRRC extension is left for future work.} The corresponding real-valued passband RF waveform can be written as follows:
\begin{align}\label{RF_Finite_Waveform}
x_{k,{\rm RF},N}(t)=\Re\left\{x_{k,N}(t)e^{{\rm j}2\pi f_{\rm c}t}\right\},
\end{align}
where $f_{\rm c}$ is the carrier frequency.

For the path from user $k$ to the BS through PA $m$, let $\alpha_{k,m}(\psi_m)\in\mathbb{R}_+$ and $\tau_{k,m}(\psi_m)\in\mathbb{R}_+$ denote the effective gain and propagation delay, respectively. Both quantities depend on the PA position. The equivalent continuous-time channel impulse response of user $k$ is given by
\begin{align}\label{Channel_Impulse}
h_k(t;{\bm\psi})=\sum_{m=1}^{M}\alpha_{k,m}(\psi_m)\delta(t-\tau_{k,m}(\psi_m)),
\end{align}
where ${\bm\psi}\triangleq[\psi_1,\ldots,\psi_M]^{\mathsf T}$ is the PA position vector.

The received RF signal at the BS front end is given by
\begin{align}\label{RF_Received_Signal_Form_1}
y_{{\rm RF},N}(t)=\sum_{k=1}^{K} h_k(t;{\bm\psi})\ast x_{k,{\rm RF},N}(t)+z_{{\rm RF},N}(t),
\end{align}
where $z_{{\rm RF},N}(t)$ denotes the real-valued passband additive white Gaussian noise (AWGN). By substituting \eqref{RF_Finite_Waveform} and \eqref{Channel_Impulse} into \eqref{RF_Received_Signal_Form_1}, we obtain
\begin{align}
&y_{{\rm RF},N}(t)=\Re\left\{\sum_{k=1}^{K}\sum_{m=1}^{M}x_{k,N}(t-\tau_{k,m}(\psi_m))\right.\nonumber\\
&\times\left.\alpha_{k,m}(\psi_m)e^{-{\rm j}2\pi f_{\rm c}\tau_{k,m}(\psi_m)}e^{{\rm j}2\pi f_{\rm c}t}\right\}+z_{{\rm RF},N}(t).\nonumber
\end{align}
After coherent downconversion and low-pass filtering, the received equivalent complex baseband signal is
\begin{align}
y_{N}(t)&=\sum_{k=1}^{K}\sum_{m=1}^{M}x_{k,N}(t-\tau_{k,m}(\psi_m))\alpha_{k,m}(\psi_m)e^{-{\rm j}2\pi f_{\rm c}\tau_{k,m}(\psi_m)}+z_{N}(t)\nonumber\\
&=\sum_{k=1}^{K}x_{k,N}(t)\ast g_k(t;{\bm\psi})+z_N(t),\label{Baseband_Receive_Signal_sub2}
\end{align}
where $z_N(t)$ is the equivalent complex baseband AWGN with two-sided power spectral density (PSD) $N_0$, and
\begin{align}\nonumber
g_k(t;{\bm\psi})\triangleq\sum_{m=1}^{M}\alpha_{k,m}(\psi_m)e^{-{\rm j}2\pi f_{\rm c}\tau_{k,m}(\psi_m)}\delta(t-\tau_{k,m}(\psi_m))
\end{align}
is the equivalent complex baseband channel impulse response of user $k$.

\subsection{Channel Model and Frequency Selectivity}
PASS is envisioned for high-frequency operation \cite{suzuki2022pinching}, for which LoS propagation is often dominant \cite{ouyang2024primer}. We therefore adopt a free-space LoS model for the over-the-air link. The path from user $k$ to the BS through PA $m$ comprises \romannumeral1) \emph{free-space propagation} from the user to the PA and \romannumeral2) \emph{guided propagation} from the PA to the feed point of its dielectric segment. The signal retains the carrier frequency $f_{\rm c}$ throughout this path, while the dielectric waveguide changes its propagation constant and propagation velocity. Let $\lambda=\frac{c}{f_{\rm c}}$ denote the free-space wavelength, where $c$ is the speed of light. We assume that the effective refractive index of the waveguide is approximately constant over the signal band. Accordingly, the propagation constant is approximately linear in frequency, and guided propagation is modeled with an approximately constant propagation velocity $\frac{c}{n_{\rm eff}}$, where $n_{\rm eff}$ denotes the effective refractive index \cite{yeh2008essence}. The propagation delay over a guided distance $L$ is therefore approximated by $\frac{n_{\rm eff}L}{c}$. Higher-order frequency dependence of the propagation constant, which gives rise to waveguide dispersion, is omitted to isolate the frequency selectivity caused by aperture-induced differential propagation delays. Such dispersion may become relevant for much wider bandwidths or operation near the waveguide cutoff and would introduce an additional source of frequency selectivity \cite{liu2025pinchingtutorial}.

Under this model, the effective gain associated with the $m$th PA for user $k$ is given by
\begin{align}
\alpha_{k,m}(\psi_m)=
\underbrace{\eta^{\frac{1}{2}}\frac{1}{\|{\mathbf u}_k-{\bm\psi}_m\|}}_{\text{free-space propagation}}\times
\underbrace{10^{-\frac{\kappa}{20}\|{\bm\psi}_m-{\bm\psi}_0^m\|}}_{\text{in-waveguide attenuation}},
\end{align}
where $\eta\triangleq \frac{c^2}{16\pi^2 f_{\rm c}^2}$ and $\kappa$ is the average dielectric-waveguide attenuation in dB/m \cite{yeh2008essence}. The total propagation delay is
\begin{align}\label{delay_model_basic}
\tau_{k,m}(\psi_m)=
\underbrace{\frac{\|{\mathbf u}_k-{\bm\psi}_m\|}{c}}_{\text{free-space delay}}
+
\underbrace{\frac{n_{\rm eff}\|{\bm\psi}_m-{\bm\psi}_0^m\|}{c}}_{\text{in-waveguide delay}}.
\end{align}

The Fourier transform of $g_k(t;{\bm\psi})$ gives the following frequency response:
\begin{align}
G_k(f;{\bm\psi})
&\triangleq {\mathcal F}(g_k(t;{\bm\psi}))=\sum_{m=1}^{M}\alpha_{k,m}(\psi_m) e^{-{\rm j}2\pi f_{\rm c}\tau_{k,m}(\psi_m)}
e^{-{\rm j}2\pi f\tau_{k,m}(\psi_m)}.\label{Frequency_Response_k}
\end{align}
Define $\tau_{k,\max}({\bm\psi})\triangleq \max_{m}\tau_{k,m}(\psi_m)$, $\tau_{k,\min}({\bm\psi})\triangleq \min_{m}\tau_{k,m}(\psi_m)$, $\Delta_{k,m}({\bm\psi})\triangleq \tau_{k,m}(\psi_m)-\tau_{k,\min}({\bm\psi})$, and the system-wide delay spread $\Delta\tau_{\rm sys}({\bm\psi})\triangleq\max_k\{\tau_{k,\max}({\bm\psi})-\tau_{k,\min}({\bm\psi})\}$. Equation \eqref{Frequency_Response_k} becomes
\begin{align}
G_k(f;{\bm\psi})
&=e^{-{\rm j}2\pi f\tau_{k,\min}({\bm\psi})}
\sum_{m=1}^{M}\alpha_{k,m}(\psi_m) e^{-{\rm j}2\pi f_{\rm c}\tau_{k,m}(\psi_m)}
e^{-{\rm j}2\pi f\Delta_{k,m}({\bm\psi})}.\nonumber
\end{align}

If $B\Delta\tau_{\rm sys}({\bm\psi})\ll1$, then $e^{-{\rm j}2\pi f\Delta_{k,m}({\bm\psi})}\approx 1$ for all $f\in[-\frac{B}{2},\frac{B}{2}]$, $k\in[K]$, and $m\in[M]$. It follows that
\begin{align}\label{Frequency_Flat_Channel_Impulse_Response}
G_k(f;{\bm\psi})
\approx
e^{-{\rm j}2\pi f\tau_{k,\min}({\bm\psi})}
h_{k,{\rm eff}}({\bm\psi}),
\end{align}
where
\begin{align}\nonumber
h_{k,{\rm eff}}({\bm\psi})\triangleq \sum_{m=1}^{M}\alpha_{k,m}(\psi_m)e^{-{\rm j}2\pi f_{\rm c}\tau_{k,m}(\psi_m)}.
\end{align}
This is the usual narrowband condition: the signal bandwidth is much smaller than the coherence bandwidth associated with the path-delay spread. Under \eqref{Frequency_Flat_Channel_Impulse_Response}, the response contains only a common linear phase term $e^{-{\rm j}2\pi f\tau_{k,\min}({\bm\psi})}$ over the signal band, and its magnitude is approximately constant. The channel is therefore FF with effective coefficient $h_{k,{\rm eff}}({\bm\psi})$. Hence,
\begin{align}\nonumber
g_k(t;{\bm\psi})={\mathcal F}^{-1}(G_k(f;{\bm\psi}))
\approx
h_{k,{\rm eff}}({\bm\psi})\delta(t-\tau_{k,\min}({\bm\psi})).
\end{align}
Accordingly, the received signal admits the following approximation:
\begin{align}\label{Frequency_Flat_Signal_Model}
y_N(t)\approx \sum_{k=1}^{K} h_{k,{\rm eff}}({\bm\psi})x_{k,N}(t-\tau_{k,\min}({\bm\psi}))+z_N(t).
\end{align}
After compensation for the user-specific delay $\tau_{k,\min}({\bm\psi})$, matched filtering and symbol-rate sampling yield the conventional FF discrete-time model. This narrowband case has been widely studied in the PASS literature \cite{liu2025pinchingtutorial,liu2025pinchingsurvey}. The approximation simplifies both performance analysis \cite{ding2024flexible,ouyang2025array} and PA placement \cite{xu2024rate,ouyang2025uplink}; closed-form or low-complexity solutions are available for several special cases.

When the waveguide spans a large region or the PAs are widely separated, $B\Delta\tau_{\rm sys}({\bm\psi})\not\ll1$ may instead hold. The approximation in \eqref{Frequency_Flat_Channel_Impulse_Response} then fails, and the channel response can vary significantly over frequency. The received signal no longer reduces to \eqref{Frequency_Flat_Signal_Model}. Both the achievable-rate analysis and the PA placement design must therefore use the exact FS model in \eqref{Baseband_Receive_Signal_sub2}.
\subsection{Discrete-Time Equivalent Model}
Under the adopted minimum-bandwidth sinc signaling model, transmit-pulse matched filtering followed by symbol-rate sampling yields a sufficient discrete-time representation of \eqref{Baseband_Receive_Signal_sub2}. Pass the received signal through $p^*(-t)$ and define $r_N(t)\triangleq y_N(t)\ast p^*(-t)$. Equations \eqref{Finite_Waveform} and \eqref{Baseband_Receive_Signal_sub2} give
\begin{align}
r_N(t)&=\sum_{k=1}^{K}\sum_{n=1}^{N}s_k[n]\sum_{m=1}^{M}e^{-{\rm j}2\pi f_{\rm c}\tau_{k,m}(\psi_m)}\alpha_{k,m}(\psi_m)
 c_p(t-nT_{\rm s}-\tau_{k,m}(\psi_m))+w_N(t),\label{eq:mf_output_expanded}
\end{align}
where $w_N(t)\triangleq z_N(t)\ast p^*(-t)$ is the filtered Gaussian noise. Since $z_N(t)$ is white with two-sided PSD $N_0$, the autocorrelation of $w_N(t)$ is ${\mathbb E}[w_N(t)w_N^*(t')]=N_0c_p(t-t')$. Sampling \eqref{eq:mf_output_expanded} at $t=\hat n T_{\rm s}$ for $\hat n\in[N]$, and defining $w_N[\hat n]\triangleq w_N(\hat n T_{\rm s})$ and $r_N[\hat n]\triangleq r_N(\hat n T_{\rm s})$, yields
\begin{align}\label{eq:sampled_received_1}
r_N[\hat n]=\sum_{k=1}^{K}\sum_{n=1}^{N}s_k[n]\hat g_k((\hat n-n)T_{\rm s};{\bm\psi})+w_N[\hat n],
\end{align}
where $\hat g_k(t;{\bm\psi})\triangleq g_k(t;{\bm\psi})\ast c_p(t)$ denotes the effective continuous-time channel pulse of user $k$, given by
\begin{align}
\hat g_k(t;{\bm\psi})=\sum_{m=1}^{M}\alpha_{k,m}(\psi_m)e^{-{\rm j}2\pi f_{\rm c}\tau_{k,m}(\psi_m)}c_p(t-\tau_{k,m}(\psi_m)).\nonumber
\end{align}
Thus, matched filtering and symbol-rate sampling convert the continuous-time segmented-waveguide channel into a discrete-time scalar linear time-invariant Gaussian channel with memory. Its effective taps are $\{\hat g_k[\ell;{\bm\psi}]\triangleq\hat g_k(\ell T_{\rm s};{\bm\psi})\}_{\ell\in\mathbb{Z}}$. Since $z_N(t)$ is Gaussian, its filtered and sampled version is jointly Gaussian. Moreover, the root-Nyquist condition gives ${\mathbb E}[w_N[n]w_N^*[n']]=N_0\delta[n-n']$. The samples are therefore uncorrelated and hence independent. Consequently, ${\mathbf w}_N\triangleq[w_N[1],\ldots,w_N[N]]^{\mathsf T}\sim{\mathcal{CN}}({\mathbf 0},N_0{\mathbf I}_N)$.

Define ${\mathbf s}_{k,N}\triangleq[s_k[1],\ldots,s_k[N]]^{\mathsf T}$ and ${\mathbf r}_N\triangleq[r_N[1],\ldots,r_N[N]]^{\mathsf T}$. Then, \eqref{eq:sampled_received_1} can be written in vector form as follows:
\begin{align}\label{finite_block_vector_model}
{\mathbf r}_N=\sum_{k=1}^{K}{\mathbf G}_{k,N}({\bm\psi}){\mathbf s}_{k,N}+{\mathbf w}_N,
\end{align}
where ${\mathbf G}_{k,N}({\bm\psi})$ is the Toeplitz matrix generated by $\{\hat g_k(\ell T_{\rm s};{\bm\psi})\}_{\ell\in\mathbb{Z}}$, with entries $[{\mathbf G}_{k,N}({\bm\psi})]_{\hat n,n}=\hat g_k((\hat n-n)T_{\rm s};{\bm\psi})$ for $\hat n,n\in[N]$. Toeplitz matrix theory \cite{gray2003asymptotic,gray2006toeplitz} characterizes the asymptotic eigenvalue distribution of ${\mathbf{G}}_{k,N}({\bm\psi})$ through its \emph{generating symbol}, or discrete-time Fourier transform, $\hat G_k({\rm e}^{{\rm j}\omega};{\bm\psi})\triangleq\sum_{\ell\in\mathbb{Z}}\hat g_k(\ell T_{\rm s};{\bm\psi})e^{{\rm j}\omega\ell}$ for $\omega\in[-\pi,\pi]$.

We assume that the transmitted symbols are independent across users and independent and identically distributed (i.i.d.) across time, with $s_k[n]\sim\mathcal{CN}(0,P_k)$, where $P_k$ denotes the average symbol energy. Hence,
${\mathbf s}_{k,N}\sim\mathcal{CN}({\mathbf 0},P_k{\mathbf I}_N)$, where the covariance matrix $P_k{\mathbf I}_N$ is Toeplitz and has the constant generating symbol
$S_{k,{\rm a}}(\omega)=P_k$ for $\omega\in[-\pi,\pi]$.

We first use this model to characterize the information-theoretic limit as $N\to\infty$. We then study finite-block CP-SC-FDE transmission and the corresponding PA-placement design.
\section{Infinite-Blocklength Analysis}\label{Section:Infinite-Blocklength Analysis}
This section assumes $N\to\infty$. We first study a single user and omit the user index $k$, then we extend the results to multiple users.
\subsection{Single-User Case}
For the single-user case, the finite-block MI is given by
\begin{align}\label{eq:finite_block_mi}
I({\mathbf s}_N;{\mathbf r}_N)=\log_2\det\left({\mathbf I}_N+{\mathbf A}_N({\bm\psi})\right),
\end{align}
where ${\mathbf A}_N({\bm\psi})\triangleq \frac{P}{N_0}{\mathbf G}_N({\bm\psi}){\mathbf G}_N^{\mathsf H}({\bm\psi})$. The asymptotic achievable rate, in bits per channel use, can be expressed as follows:
\begin{align}\nonumber
{\mathsf R}({\bm\psi})=\lim_{N\rightarrow\infty}\frac{1}{N}I({\mathbf s}_N;{\mathbf r}_N).
\end{align}
It remains to characterize the normalized log-determinant in \eqref{eq:finite_block_mi}. The main result is summarized in Theorem \ref{thm:single_user_asymptotic_rate}.
\vspace{-10pt}
\begin{theorem}\label{thm:single_user_asymptotic_rate}
Under the adopted signaling model, the asymptotic achievable rate is given by
\begin{align}\nonumber
R({\bm\psi})
=
\frac{1}{B}
\int_{-\frac{B}{2}}^{\frac{B}{2}}
\log_2\left(
1+\gamma_0|G(f;{\bm\psi})|^2
\right){\rm d}f ,
\end{align}
where $\gamma_0\triangleq\frac{P}{N_0}$ denotes the signal-to-noise ratio (SNR).
\end{theorem}
\vspace{-10pt}
\begin{IEEEproof}
See Appendix \ref{proof_thm:single_user_asymptotic_rate}.
\end{IEEEproof}
\subsubsection{PA Placement Optimization}
The rate ${\mathsf R}({\bm\psi})$ depends on the PA positions through the geometry-dependent gains and delays. The PA placement problem is formulated as follows:
\begin{align}\label{problem_single_carrier_capacity}
\max_{{\bm\psi}\in{\mathcal X}}~{\mathsf R}({\bm\psi}),
\end{align}
where the feasible set is given by
\begin{align}\nonumber
{\mathcal X}\triangleq\left\{{\bm\psi}\left|
\begin{matrix}
\psi_m\in[\psi_0^m,\psi_0^m+L], m\in[M],\\
|\psi_m-\psi_{m'}|\geq\Delta, m\neq m'
\end{matrix}
\right.\right\}.
\end{align}

We solve the nonconvex problem in \eqref{problem_single_carrier_capacity} through element-wise alternating optimization. Each step optimizes $\psi_m$ with all other PA positions fixed:
\begin{align}
\max_{\psi_m}~&{\mathsf R}^{(m)}(\psi_m)
\triangleq
\frac{1}{B}\int_{-\frac{B}{2}}^{\frac{B}{2}}
\log_2(1+\gamma_0|G^{(m)}(f,\psi_m)|^2){\rm d}f\nonumber\\
 {\rm s.t.}~&\psi_m\in[\psi_0^m,\psi_0^m+L],\ |\psi_m-\psi_{m'}|\geq\Delta,\ m\neq m',\nonumber
\end{align}
where $G^{(m)}(f,\psi_m)\triangleq \tilde{G}^{(m)}(f)+\alpha_m(\psi_m)e^{-{\rm j}2\pi(f_{\rm c}+f)\tau_m(\psi_m)}$ and $\tilde{G}^{(m)}(f)\triangleq\sum_{m'\neq m}\alpha_{m'}(\psi_{m'})e^{-{\rm j}2\pi(f_{\rm c}+f)\tau_{m'}(\psi_{m'})}$. Because the objective depends on both the free-space phase and the guided-wave phase, ${\mathsf R}^{(m)}(\psi_m)$ is highly oscillatory in space \cite{wang2025modeling,xu2026pinching,liu2025pinchingtutorial}. As a result, standard gradient-based methods are prone to poor local optima. We therefore adopt a one-dimensional grid search over the feasible interval. Specifically, the interval $[\psi_0^m,\psi_0^m+L]$ is discretized into a uniform $Q$-point grid:
\begin{align}\nonumber
{\mathcal Q}_m\triangleq\left\{
\psi_0^m,\psi_0^m+\frac{L}{Q-1},\ldots,\psi_0^m+L
\right\}.
\end{align}
To enforce the minimum-spacing constraint, define the invalid set as follows:
\begin{align}\nonumber
\hat{{\mathcal Q}}_m\triangleq\{x\in{\mathcal Q}_m:|x-\psi_{m'}|<\Delta, m'\neq m\}. 
\end{align}
The updated PA position is then chosen as follows:
\begin{align}\nonumber
\psi_m^\star\triangleq\argmax\nolimits_{\psi_m\in{\mathcal Q}_m\setminus\hat{{\mathcal Q}}_m}{\mathsf R}^{(m)}(\psi_m).
\end{align}
Algorithm \ref{Algorithm1} updates the PAs sequentially until the objective no longer improves. Since each update selects the best feasible grid point and the finite-grid objective is nondecreasing, the algorithm terminates at a coordinate-wise optimum on the adopted grid. Its complexity is ${\mathcal O}(I_{\rm iter}MQI_{\rm GL})$ when each rate evaluation uses $I_{\rm GL}$ quadrature points. We approximate the frequency integral by Gauss-Legendre quadrature:
\begin{align}
{\mathsf R}({\bm\psi})
\approx
\frac{1}{2}\sum_{i=1}^{I_{\rm GL}}\varpi_i
\log_2\left(
1+\gamma_0|G(B\xi_i/2;{\bm\psi})|^2
\right),\nonumber
\end{align}
where $\{\varpi_i\}$ and $\{\xi_i\}$ denote the standard Gauss-Legendre weights and abscissas, respectively. The number of quadrature points $I_{\rm GL}$ controls the accuracy-complexity tradeoff.

\begin{algorithm}[!t]
\caption{Element-wise algorithm for solving \eqref{problem_single_carrier_capacity}}
\label{Algorithm1}
\begin{algorithmic}[1]
\STATE Initialize ${\bm\psi}$
\REPEAT
\FOR{$m\in\{1,\ldots,M\}$}
\STATE Update $\psi_m$ via one-dimensional search
\ENDFOR
\UNTIL{convergence}
\end{algorithmic}
\end{algorithm}
\subsubsection{Comparison With the FF Case}
We next quantify the error caused by applying the conventional FF approximation to the FS PASS channel. 

Define $a_m\triangleq\alpha_m(\psi_m)e^{-{\rm j}2\pi f_{\rm c}\tau_m(\psi_m)}$. The exact frequency response is $G(f;{\bm\psi})=\sum_{m=1}^{M}a_me^{-{\rm j}2\pi f\tau_m(\psi_m)}$. We temporarily omit the dependence on ${\bm\psi}$ and write $G(f)=\sum_{m=1}^{M}a_me^{-{\rm j}2\pi f\tau_m}$. The exact FS achievable rate is ${\mathsf R}_{\rm FS}\triangleq\frac{1}{B}\int_{-B/2}^{B/2}\log_2(1+\gamma_0|G(f)|^2){\rm d}f$. Under the FF approximation, $G(f)\approx e^{-{\rm j}2\pi f\tau_{\min}}h_{\rm eff}$, where $h_{\rm eff}=\sum_{m=1}^{M}a_m$, and the corresponding FF rate is $\mathsf R_{\rm FF}\triangleq\log_2(1+\gamma_0|h_{\rm eff}|^2)$.
Define the rate overestimation as follows:
\begin{align}\nonumber
\Delta\mathsf R
\triangleq
\mathsf R_{\rm FF}-\mathsf R_{\rm FS}=
\frac{1}{B}
\int_{-\frac{B}{2}}^{\frac{B}{2}}
\log_2\left(
\frac{
1+\gamma_0|h_{\rm eff}|^2
}{
1+\gamma_0|G(f)|^2
}
\right){\rm d}f.
\end{align}

Suppose that the PA positions align all path components at the carrier frequency, i.e., $a_m=b_me^{{\rm j}\phi_0}$ with $b_m>0$. This placement structure is common under the FF PASS model \cite{xu2024rate,ouyang2025uplink,ouyang2025array}. Define $A\triangleq\sum_{m=1}^{M}b_m$ and $w_m\triangleq\frac{b_m}{A}$, so $\sum_{m=1}^{M}w_m=1$. Then, $h_{\rm eff}=Ae^{{\rm j}\phi_0}$ and $|h_{\rm eff}|^2=A^2$. By defining $\rho\triangleq\gamma_0A^2$, we obtain $\mathsf R_{\rm FF}=\log_2(1+\rho)$. Define the amplitude-weighted mean delay and centered delays as $\overline{\tau}\triangleq\sum_{m=1}^{M}w_m\tau_m$ and $\delta_m\triangleq\tau_m-\overline{\tau}$, respectively. By construction, $\sum_{m=1}^{M}w_m\delta_m=0$. The exact frequency response is $G(f)=Ae^{{\rm j}\phi_0}e^{-{\rm j}2\pi f\overline{\tau}}\Phi_\tau(2\pi f)$, where $\Phi_\tau(\omega)\triangleq\sum_{m=1}^{M}w_me^{-{\rm j}\omega\delta_m}$. Consequently, ${\mathsf{R}}_{\rm FS}=\frac{1}{B}\int_{-B/2}^{B/2}\log_2(1+\rho|\Phi_{\tau}(2\pi f)|^2){\rm d}f$.

Since $|\Phi_\tau(2\pi f)|=\lvert\sum_{m=1}^{M}w_me^{-{\rm j}2\pi f\delta_{m}}\rvert\leq\sum_{m=1}^{M}w_m=1$, we have $|G(f)|^2\leq A^2$ and therefore $\mathsf R_{\rm FS}\leq\mathsf R_{\rm FF}$. Equality over a nonzero frequency interval holds only when all paths have the same propagation delay.
\vspace{-10pt}
\begin{remark}
Under carrier-phase-aligned PA placement, the FF rate is an upper bound on the exact FS rate. The bound is strict whenever the active PA-assisted paths have unequal delays.
\end{remark}
\vspace{-10pt}
The above result establishes that the FF channel model overestimates the exact FS rate, but it does not quantify how the FF rate overestimation evolves with the aperture-induced delay spread. An exact closed-form characterization of $\Delta\mathsf R$ for arbitrary delay profiles and delay spreads is generally unavailable. We therefore develop two complementary characterizations. First, we consider small but nonzero delay spreads to quantify how rapidly the FF approximation loses accuracy as the channel departs from the FF condition. Second, we derive a lower bound on $\Delta\mathsf R$ to characterize how the FF rate overestimation grows with the waveguide aperture, or equivalently with the aperture-induced delay spread.
\subsubsection{Small-Delay-Spread Characterization} 
As discussed above, increasing the waveguide aperture can enlarge the differential propagation delays among the PA-assisted paths and eventually violate the FF condition $B\Delta\tau\ll1$, where $\Delta\tau\triangleq\tau_{\max}-\tau_{\min}$ represents the aperture-induced delay spread. To quantify how rapidly the FF approximation begins to lose accuracy as $\Delta\tau$ increases from zero, we first consider small but nonzero delay spreads for which $B\Delta\tau\ll1$ still holds. Define the amplitude-weighted central delay moments as $\mu_q\triangleq\sum_{m=1}^{M}w_m\delta_m^q$ for $q\geq1$. In particular, $\mu_1=0$, while $\mu_2$ represents the amplitude-weighted delay variance.
\vspace{-10pt}
\begin{theorem}\label{them:rate_difference}
Under carrier-phase-aligned placement and for a fixed bandwidth $B$, as $\Delta\tau\rightarrow0$ (equivalently, $B\Delta\tau\rightarrow0$), the FF rate overestimation satisfies
\begin{equation}\label{rate_gap_second_order}
\Delta{\mathsf{R}}=
\frac{B^2\pi^2}{3\ln 2}\frac{\rho}{1+\rho}\mu_2
+{\mathcal{O}}((B\Delta\tau)^4).
\end{equation}
\end{theorem}
\vspace{-10pt}
\begin{IEEEproof}
See Appendix \ref{proof_them:rate_difference}.
\end{IEEEproof}
To make the dependence on $\Delta\tau$ explicit, define $\nu_2\triangleq\frac{\mu_2}{\Delta\tau^2}$ for $\Delta\tau>0$. Since the weighted delay distribution is supported on $[\tau_{\min},\tau_{\max}]$, Popoviciu's inequality gives $\mu_2\leq\frac{\Delta\tau^2}{4}$ and hence $0<\nu_2\leq\frac{1}{4}$. For a fixed bandwidth $B$, \eqref{rate_gap_second_order} can therefore be expressed as follows:
\begin{equation}
\Delta\mathsf R
=
\frac{B^2\pi^2}{3\ln 2}
\frac{\rho}{1+\rho}
\nu_2(\Delta\tau)^2
+
\mathcal O((\Delta\tau)^4).
\label{rate_gap_quadratic_onset}
\end{equation}
The coefficient in \eqref{rate_gap_quadratic_onset} depends on the received SNR via $\frac{\rho}{1+\rho}$ and on the normalized weighted delay profile via $\nu_2$. 
\vspace{-10pt}
\begin{remark}
For a fixed bandwidth $B$ and sufficiently small delay spread $\Delta\tau$, \eqref{rate_gap_quadratic_onset} shows that the FF rate overestimation exhibits an approximately quadratic dependence on $\Delta\tau$. This result quantifies how rapidly the FF approximation deteriorates as the waveguide aperture increases and the resulting aperture-induced delay spread grows within the small-delay-spread regime.
\end{remark}
\vspace{-10pt}
We next relate the FF rate overestimation to the PASS geometry. Referring to \eqref{delay_model_basic}, the delay of the path associated with the $m$th segment is $\tau_m=\frac{r_m+n_{\rm eff}\ell_m}{c}$, where $r_m\triangleq\|{\mathbf u}-{\bm\psi}_m\|$ and $\ell_m\triangleq\|{\bm\psi}_m-{\bm\psi}_0^m\|$ denote the free-space and in-waveguide propagation distances, respectively. Define their amplitude-weighted means as $\overline r\triangleq\sum_{m=1}^{M}w_mr_m$ and $\overline\ell\triangleq\sum_{m=1}^{M}w_m\ell_m$. The centered delay can then be written as $\delta_m=\frac{(r_m-\overline r)+n_{\rm eff}(\ell_m-\overline\ell)}{c}$. Accordingly, the amplitude-weighted delay variance is $\mu_2=\frac{1}{c^2}(
\sigma_r^2
+
n_{\rm eff}^2\sigma_\ell^2
+
2n_{\rm eff}\sigma_{r,\ell}
)$, which yields
\begin{align}
\Delta\mathsf R
\approx
\frac{\pi^2B^2\rho}
{3c^2(1+\rho)\ln 2}
\left(
\sigma_r^2
+
n_{\rm eff}^2\sigma_\ell^2
+
2n_{\rm eff}\sigma_{r,\ell}
\right).\nonumber
\end{align}
where $\sigma_r^2\triangleq\sum_{m=1}^{M}w_m(r_m-\overline r)^2$, $\sigma_\ell^2\triangleq\sum_{m=1}^{M}w_m(\ell_m-\overline\ell)^2$, and $\sigma_{r,\ell}\triangleq\sum_{m=1}^{M}w_m(r_m-\overline r)(\ell_m-\overline\ell)$. 

Here, $\sigma_r^2$ characterizes the variation of the free-space propagation distances across the PA-assisted paths, $n_{\rm eff}^2\sigma_\ell^2$ characterizes the corresponding variation of the guided propagation distances, and $2n_{\rm eff}\sigma_{r,\ell}$ captures their interaction. Hence, the PASS geometry determines not only carrier-phase alignment and received power but also the delay dispersion that governs the local accuracy of the FF approximation.
\subsubsection{Large-Delay-Spread Characterization}
As the waveguide aperture continues to increase, the aperture-induced delay spread $\Delta\tau$ can become large enough that the FF condition $B\Delta\tau\ll1$ no longer holds. The preceding small-delay-spread characterization then ceases to apply. To characterize the FF rate overestimation beyond this regime, we derive a lower bound on $\Delta\mathsf R$ that is valid for arbitrary $B\Delta\tau$ and use it to investigate the large-delay-spread behavior. Since $\Delta\tau$ increases with the waveguide aperture in the considered PASS geometry, the resulting analysis also characterizes the large-aperture behavior.

Define the normalized band-averaged channel power as follows:
\begin{align}
\mathcal Q_B
\triangleq\frac{1}{B}\int_{-\frac{B}{2}}^{\frac{B}{2}}
|\Phi_\tau(2\pi f)|^2{\rm d}f
=\sum_{m=1}^{M}w_m^2+\mathcal C_B,
\nonumber
\end{align}
where $\mathcal C_B\triangleq\sum_{m\neq m'}w_mw_{m'}\frac{\sin\left(\pi B(\tau_m-\tau_{m'})\right)}{\pi B(\tau_m-\tau_{m'})}$. Since $|\Phi_\tau(2\pi f)|\leq1$, we have $0\leq\mathcal Q_B\leq1$. Jensen's inequality yields $\mathsf R_{\rm FS}\leq\log_2(1+\rho\mathcal Q_B)$, and it follows that
\begin{equation}
\Delta\mathsf R\geq\log_2\left(\frac{1+\rho}{1+\rho\mathcal Q_B}\right).
\label{rate_gap_Jensen_bound}
\end{equation}

Assume that the active paths have pairwise-distinct delays, i.e., $\tau_m\ne\tau_{m'}$ for $m\ne m'$ \cite{ouyang2025array,ouyang2025uplink}. Carrier-phase alignment gives $e^{-{\rm j}2\pi f_{\rm c}\tau_m}=e^{{\rm j}\phi_0}$ for every $m$. Hence, the values $q_m\triangleq f_{\rm c}(\tau_m-\tau_1)$ are pairwise-distinct integers, and $\tau_m-\tau_{m'}=\frac{q_m-q_{m'}}{f_{\rm c}}$. Let $z_m^{(\pm)}\triangleq w_m e^{\pm{\rm j}\pi B\tau_m}$. The identity $\sin(\pi B(\tau_m-\tau_{m'}))=\frac{e^{{\rm j}\pi B(\tau_m-\tau_{m'})}-e^{-{\rm j}\pi B(\tau_m-\tau_{m'})}}{2{\rm j}}$ gives $\mathcal C_B=\frac{\mathcal H_+-\mathcal H_-}{2\pi{\rm j}B}$, where $\mathcal H_\pm\triangleq f_{\rm c}\sum_{m\neq m'}\frac{z_m^{(\pm)}(z_{m'}^{(\pm)})^*}{q_m-q_{m'}}$. Because the $q_m$ are distinct integers, the classical Hilbert inequality gives $|\mathcal H_\pm|\leq\pi f_{\rm c}\sum_mw_m^2$ \cite{steele2004cauchy}. Therefore,
\begin{equation}
|\mathcal C_B|\leq\frac{|\mathcal H_+|+|\mathcal H_-|}{2\pi B}\leq\frac{f_{\rm c}}{B}\sum_{m=1}^{M}w_m^2.
\label{cross_term_bound}
\end{equation}

\begin{figure}[!t]
\centering
    \subfigure[Centered user.]
    {
        \includegraphics[width=0.45\textwidth]{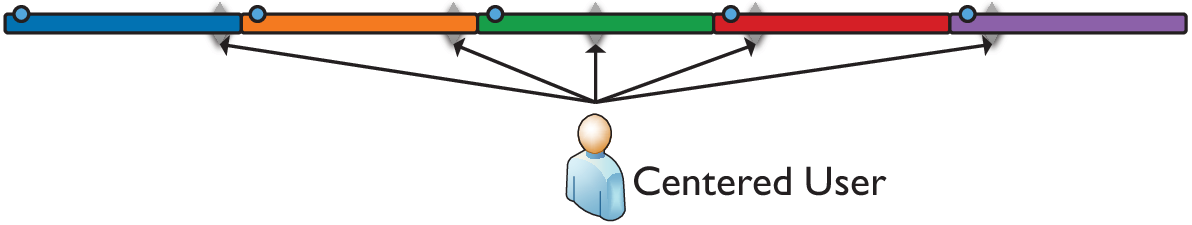}
        \label{Figure: User_Center}
    }
    \subfigure[Edge user.]
    {
        \includegraphics[width=0.45\textwidth]{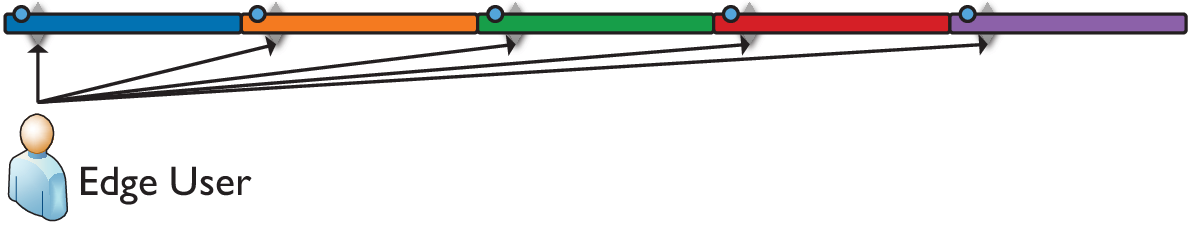}
        \label{Figure: User_Edge}
    }
\caption{Illustration of user locations.}
\label{Figure: User_Deployment}
\vspace{-15pt}
\end{figure}

Define the effective number of contributing paths as follows:
\begin{equation}
M_{\rm eff}\triangleq\frac{1}{\sum_{m=1}^{M}w_m^2}
=\frac{(\sum_{m=1}^{M}b_m)^2}{\sum_{m=1}^{M}b_m^2}.
\nonumber
\end{equation}
We have $1\leq M_{\rm eff}\leq M$, with $M_{\rm eff}=M$ for equal path amplitudes. Equation \eqref{cross_term_bound}, together with $\mathcal Q_B\leq1$, yields
\begin{equation}
\mathcal Q_B\leq\min\left\{1,\frac{1+f_{\rm c}/B}{M_{\rm eff}}\right\}.\nonumber
\end{equation}
Substitution into \eqref{rate_gap_Jensen_bound} gives the finite-bandwidth, finite-aperture lower bound as follows:
\begin{equation}
\Delta\mathsf R\geq\log_2\left[
\frac{1+\rho}{1+\rho\min\{1,(1+f_{\rm c}/B)/M_{\rm eff}\}}
\right].
\nonumber
\end{equation}
The bound reduces to $\Delta\mathsf R\geq0$ when $M_{\rm eff}\leq1+f_{\rm c}/B$. When $M_{\rm eff}>1+f_{\rm c}/B$, it becomes
\begin{equation}
\Delta\mathsf R\geq\log_2\left(\frac{1+\rho}{1+(1+f_{\rm c}/B)\rho/M_{\rm eff}}\right).
\label{rate_gap_effective_paths}
\end{equation}
We next evaluate $M_{\rm eff}$ for the two representative user locations in {\figurename} {\ref{Figure: User_Deployment}}. The user lies directly below the waveguide axis, so its perpendicular distance to the waveguide is $d$.
\begin{itemize}
  \item \emph{Centered User:} Let $M=2J+1$ and index the central segment as $m_0=J+1$. Following the carrier-phase-aligned placement in \cite{ouyang2025array,ouyang2025uplink}, the PA in each segment is first placed near the endpoint closest to the user and then adjusted over a wavelength-scale distance to align its carrier phase. Specifically, a right-side PA is shifted rightward from its feed point, whereas a left-side PA is shifted leftward from the endpoint closest to the user toward its feed point. These wavelength-scale adjustments have a negligible effect on the free-space path amplitudes. By also neglecting the in-waveguide attenuation over each short segment \cite{ouyang2025uplink}, the path amplitudes can be approximated as $b_{m_0}\approx\frac{\sqrt{\eta}}{d}$ and $b_{m_0\pm q}\approx\frac{\sqrt{\eta}}{({d^2+(q-\frac{1}{2})^2L^2})^{1/2}}$ for $q=1,\ldots,J$. It follows that
\begin{equation}
M_{{\rm eff},{\rm c}}
\approx
\frac{
\left[
\frac{1}{d}
+
2\sum_{q=1}^{J}
\frac{1}
{({d^2+(q-\frac{1}{2})^2L^2})^{1/2}}
\right]^2
}{
\frac{1}{d^2}
+
2\sum_{q=1}^{J}
\frac{1}
{d^2+(q-\frac{1}{2})^2L^2}
}.
\label{center_user_effective_paths}
\end{equation}
  \item \emph{Edge User:} For a user at the left edge, all segments lie to its right. The minimum-loss reference position in segment $m$ is its left feed point, followed by a wavelength-scale phase adjustment. The path amplitudes are approximated by $b_m\approx\frac{\sqrt{\eta}}{({d^2+(m-1)^2L^2})^{1/2}}$ for $m=1,\ldots,M$. Therefore,
\begin{equation}
M_{{\rm eff},{\rm e}}
\approx
\frac{
\left[
\sum_{m=1}^{M}
\frac{1}
{({d^2+(m-1)^2L^2})^{1/2}}
\right]^2
}{
\sum_{m=1}^{M}
\frac{1}
{d^2+(m-1)^2L^2}
}.
\label{edge_user_effective_paths}
\end{equation}
\end{itemize}

Substituting \eqref{center_user_effective_paths} and \eqref{edge_user_effective_paths} into \eqref{rate_gap_effective_paths} gives, for $i\in\{{\rm c},{\rm e}\}$,
\begin{equation}
\Delta\mathsf R_i
\geq
\log_2
\left(
\frac{1+\rho_i}
{1+(1+f_{\rm c}/B)\rho_i/M_{{\rm eff},i}}
\right)\triangleq
\tilde{\Delta}\mathsf R_i,\label{rate_gap_effective_paths_case}
\end{equation}
where $\rho_i=\gamma_0(\sum_{m=1}^{M}b_{m,i})^2$. We next use the sum-to-integral approximations in \cite{ouyang2025array,ouyang2025uplink} to obtain tractable finite-aperture expressions. For the centered and edge users, respectively,
\begin{subequations}\label{effective_paths_integral_approx}
\begin{align}
M_{{\rm eff},{\rm c}}&\approx
\frac{\left[d^{-1}+\frac{2}{L}\operatorname{asinh}\!\left(\frac{D_x-L}{2d}\right)\right]^2}
{d^{-2}+\frac{2}{dL}\arctan\!\left(\frac{D_x-L}{2d}\right)},\\
M_{{\rm eff},{\rm e}}&\approx
\frac{\left[d^{-1}+\frac{1}{L}\operatorname{asinh}\!\left(\frac{D_x-L}{d}\right)\right]^2}
{d^{-2}+\frac{1}{dL}\arctan\!\left(\frac{D_x-L}{d}\right)},
\end{align}
\end{subequations}
and the corresponding center-frequency SNRs are
\begin{subequations}\label{effective_snr_integral_approx}
\begin{align}
\rho_{
\rm c}&\approx\gamma_0\eta
\left[d^{-1}+\frac{2}{L}\operatorname{asinh}\!\left(\frac{D_x-L}{2d}\right)\right]^2,\\
\rho_{
\rm e}&\approx\gamma_0\eta
\left[d^{-1}+\frac{1}{L}\operatorname{asinh}\!\left(\frac{D_x-L}{d}\right)\right]^2.
\end{align}
\end{subequations}
These expressions directly approximate the lower bound in \eqref{rate_gap_effective_paths_case}. They also expose its large-aperture behavior. For fixed $L$ and $d$, the amplitude sum grows logarithmically with $D_x$, whereas the squared-amplitude sum approaches a finite positive value. Hence, as ${D_x}/{d}\rightarrow\infty$,
\begin{align}
M_{{\rm eff},{\rm c}}
\sim
\frac{4d}{\frac{L^2}{d}+\pi L}
\ln^2
\left(
\frac{D_x}{d}
\right),
M_{{\rm eff},{\rm e}}
\sim
\frac{2d}{\pi L}
\ln^2
\left(
\frac{2D_x}{d}
\right).
\nonumber
\end{align}
The same amplitude-sum approximation gives
\begin{align}
\rho_{\rm c}
\sim
\frac{4\gamma_0\eta}{L^2}
\ln^2
\left(
\frac{D_x}{d}
\right),\quad
\rho_{\rm e}
\sim
\frac{\gamma_0\eta}{L^2}
\ln^2
\left(
\frac{2D_x}{d}
\right).
\label{edge_rate_gap_lower_scaling1}
\end{align}
Consequently,
\begin{align}
\frac{\rho_{\rm c}}{M_{{\rm eff},{\rm c}}}
\rightarrow
\frac{(\pi+L/d)\gamma_0\eta}{dL},\quad
\frac{\rho_{\rm e}}{M_{{\rm eff},{\rm e}}}
\rightarrow
\frac{\pi\gamma_0\eta}{2dL}.
\label{rate_gap_upper_scaling2}
\end{align}

For every fixed $B>0$, the factor $1+f_{\rm c}/B$ and the limits in
\eqref{rate_gap_upper_scaling2} are finite. Moreover,
$M_{{\rm eff},i}\rightarrow\infty$, so
$M_{{\rm eff},i}>1+f_{\rm c}/B$ for all sufficiently large $D_x$.
Combining \eqref{rate_gap_effective_paths_case},
\eqref{edge_rate_gap_lower_scaling1}, and
\eqref{rate_gap_upper_scaling2} yields
\begin{equation}
\Delta\mathsf R_i\geq \tilde{\Delta}\mathsf R_i
\sim
2\log_2\ln
\left(
\frac{D_x}{d}
\right),
\quad
i\in\{{\rm c},{\rm e}\},
\label{rate_gap_aperture_scaling}
\end{equation}
as $D_x/d\rightarrow\infty$, for every fixed $B>0$. To express this scaling directly in terms of the aperture-induced delay spread, let $\Delta\tau_i$ denote the delay spread for user location $i\in\{{\rm c},{\rm e}\}$. For the centered and edge users, $\Delta\tau_{\rm c}\sim\frac{D_x}{2c}$ and $\Delta\tau_{\rm e}\sim\frac{D_x}{c}$, respectively, as $D_x/d\rightarrow\infty$. Hence, \eqref{rate_gap_aperture_scaling} can equivalently be written as follows:
\begin{equation}
\Delta\mathsf R_i\geq \tilde{\Delta}\mathsf R_i
\sim
2\log_2\ln
\left(
\frac{c\Delta\tau_i}{d}
\right),
\quad
i\in\{{\rm c},{\rm e}\},
\label{rate_gap_aperture_scaling_tau}
\end{equation}
as $c\Delta\tau_i/d\rightarrow\infty$, for every fixed $B>0$.
\vspace{-10pt}
\begin{remark}
As the waveguide aperture and the resulting aperture-induced delay spread become large, \eqref{rate_gap_aperture_scaling_tau} shows that the derived lower bound grows double logarithmically with $\Delta\tau_i$. Consequently, the exact FF rate overestimation grows at least double logarithmically with the aperture-induced delay spread.
\end{remark}
\vspace{-10pt}
\subsection{Multiuser Case}
We now extend the analysis to the multiuser uplink in \eqref{finite_block_vector_model}. We consider joint and parallel decoding and derive their achievable sum-rates.
\subsubsection{Joint Decoding}
Under joint decoding, the receiver jointly decodes all $K$ messages. The asymptotic sum-rate is the MI between all transmitted sequences and the received sequence \cite{verdu1989multiple}:
\begin{align}
{\mathsf{R}}_{\rm J}({\bm\psi})&=\lim_{N\rightarrow\infty}\frac{1}{N}I({\mathbf{s}}_{1,N},\ldots,{\mathbf{s}}_{K,N};{\mathbf{r}}_N)\nonumber\\
&=\lim_{N\rightarrow\infty}\frac{1}{N}\log_2\det\left({\mathbf{I}}_N+\sum_{k=1}^{K}{{\mathbf{A}}_{k,N}({\bm\psi})}
\right),\nonumber
\end{align}
where ${\mathbf{A}}_{k,N}({\bm\psi})\triangleq\frac{P_k}{N_0}{\mathbf{G}}_{k,N}({\bm\psi}){\mathbf{G}}_{k,N}^{\mathsf{H}}({\bm\psi})$. The Toeplitz-symbol argument used for Theorem \ref{thm:single_user_asymptotic_rate} gives
\begin{equation}\label{joint_infinite_blocklength_rate_limit}
{\mathsf R}_{\rm J}({\bm\psi})
=\frac{1}{B}\int_{-\frac{B}{2}}^{\frac{B}{2}}
\log_2\left(
1+\sum_{k=1}^{K}A_{k}(f;{\bm{\psi}})
\right){\rm d}f,
\end{equation}
where $A_{k}(f;{\bm{\psi}})\triangleq \frac{P_k}{N_0}|G_k(f;{\bm\psi})|^2$. A detailed derivation of \eqref{joint_infinite_blocklength_rate_limit} is provided in Appendix \ref{proof_asym_joint_decoding}. Joint decoding generally incurs complexity that grows exponentially with the number of users. Successive interference cancellation (SIC) provides an alternative implementation with much lower complexity and achieves the same sum-rate as joint decoding \cite{heath2018foundations}.
\subsubsection{Parallel Decoding}
Under parallel decoding, user $k$ is decoded while the other users are treated as colored Gaussian interference \cite{hassibi2003much}. Its asymptotic achievable rate is
\begin{align}
{\mathsf R}_{k,{\rm P}}({\bm\psi})
&=
\lim_{N\rightarrow\infty}\frac{1}{N}I({\mathbf s}_{k,N};{\mathbf r}_N)\nonumber\\
&=
\lim_{N\rightarrow\infty}\frac{1}{N}
\log_2\det\left(
{\mathbf I}_N+{\mathbf{A}}_{k,N}({\bm\psi}){\mathbf Z}_{k,N}^{-1}
\right),\nonumber
\end{align}
where ${\mathbf Z}_{k,N}\triangleq {\mathbf I}_N+\sum_{k'\neq k}{\mathbf{A}}_{k',N}({\bm\psi})$ is the interference-plus-noise covariance matrix. The Toeplitz argument used for \eqref{joint_infinite_blocklength_rate_limit} yields the asymptotic sum-rate ${\mathsf R}_{\rm P}({\bm\psi})\triangleq\sum_{k=1}^{K}{\mathsf R}_{k,{\rm P}}({\bm\psi})$:
\begin{align}\label{eq:parallel_rate_discrete}
{\mathsf R}_{\rm P}({\bm\psi})
=
\sum_{k=1}^{K}\frac{1}{B}\int_{-\frac{B}{2}}^{\frac{B}{2}}
\log_2\left(
1+
\frac{A_{k}(f;{\bm{\psi}})}
{1+\hat{A}_{k}(f;{\bm{\psi}})}
\right){\rm d}f,
\end{align}
where $\hat{A}_{k}(f;{\bm{\psi}})\triangleq\sum_{k'\neq k}A_{k'}(f;{\bm{\psi}})$. A detailed derivation of \eqref{eq:parallel_rate_discrete} is presented in Appendix \ref{proof_asym_parallel_decoding}. 
\subsubsection{PA Placement Optimization}
Both ${\mathsf R}_{\rm J}({\bm\psi})$ and ${\mathsf R}_{\rm P}({\bm\psi})$ depend explicitly on the common PA position vector. The multiuser PA placement problem follows from \eqref{problem_single_carrier_capacity} after its objective is replaced by either sum-rate. The same element-wise framework applies.

\section{Finite-Block CP-SC-FDE}\label{Section: Finite-Block CP-SC-FDE}
The preceding section characterizes the asymptotic information-theoretic benchmark as the observation and coding duration become arbitrarily large, i.e., $N\rightarrow\infty$. Since CP-SC-FDE provides a practical framework for handling frequency selectivity through frequency-domain equalization \cite{falconer2002frequency,pancaldi2008single}, we now adopt it for transmission over the FS PASS channel with a finite processing-block size. Each processing block contains $N$ data symbols and is preceded by a CP of finite length.\footnote{Here, $N$ denotes the number of data symbols in each CP-SC-FDE processing block and should be distinguished from the total codeword length. A channel code may span multiple SC-FDE blocks. Specifically, if a codeword spans $J$ blocks, the total number of coded data symbols is $n_{\rm c}=JN$. The analysis in this section keeps $N$ finite to account for the CP overhead, the finite DFT grid, and the resulting SC-FDE equalization effects, while the achievable-rate interpretation allows $J\rightarrow\infty$ and hence $n_{\rm c}\rightarrow\infty$. Therefore, the derived rates characterize a fixed-$N$ CP-SC-FDE architecture rather than the finite-codeword reliability regime. If $n_{\rm c}$ is also finite and a prescribed decoding error probability is imposed, an additional finite-codeword penalty associated with channel dispersion should be taken into account \cite{polyanskiy2010channel}.} A key contribution of this section is to characterize the achievable rate of PASS under this practical transmission framework, reveal its dependence on PA placement, and develop an FS-aware placement method to improve the achievable rate.

For a strictly bandlimited waveform, the matched pulse $c_p(t)$ generally has infinite temporal support, so the sampled channel can contain infinitely many taps. Practical block transmission truncates the weak pulse tails to obtain a finite-memory approximation. Let $T_{\rm p}$ denote the retained half-duration of $c_p(t)$, chosen to satisfy $\int_{-T_{\rm p}}^{T_{\rm p}}|c_p(t)|^2{\rm d}t\geq(1-\epsilon_{\rm p})\int_{-\infty}^{\infty}|c_p(t)|^2{\rm d}t$, where $\epsilon_{\rm p}\in(0,1)$ is the truncation tolerance. We use this approximation to derive the spectral efficiency of CP-SC-FDE with symbolwise demapping \cite{falconer2002frequency,pancaldi2008single}.
\subsection{Single-User Case}
For notational simplicity, we drop the user index. According to \eqref{eq:mf_output_expanded}, the $m$th PA-assisted path contributes $\alpha_m(\psi_m)e^{-{\rm j}2\pi f_{\rm c}\tau_m(\psi_m)}c_p(t-\tau_m(\psi_m))$. After truncating $c_p(t)$ to $[-T_{\rm p},T_{\rm p}]$, the retained response of this path is supported over $[\tau_m(\psi_m)-T_{\rm p},\tau_m(\psi_m)+T_{\rm p}]$. The first and last retained sample indices are given by
\begin{align}
\ell_{\min}({\bm\psi})
&=
\min_{m\in[M]}
\left\lceil
\frac{\tau_m(\psi_m)-T_{\rm p}}{T_{\rm s}}
\right\rceil
=
\left\lceil
\frac{\tau_{\min}({\bm\psi})-T_{\rm p}}{T_{\rm s}}
\right\rceil,
\nonumber\\
\ell_{\max}({\bm\psi})
&=
\max_{m\in[M]}
\left\lfloor
\frac{\tau_m(\psi_m)+T_{\rm p}}{T_{\rm s}}
\right\rfloor
=
\left\lfloor
\frac{\tau_{\max}({\bm\psi})+T_{\rm p}}{T_{\rm s}}
\right\rfloor.
\nonumber
\end{align}
respectively.
The retained discrete-time channel is
$\{\hat g[\ell;{\bm\psi}]\}_{\ell=\ell_{\min}({\bm\psi})}^{\ell_{\max}({\bm\psi})}$,
and its memory is
$L_{\rm m}({\bm\psi})\triangleq
\ell_{\max}({\bm\psi})-\ell_{\min}({\bm\psi})$. Up to the integer-rounding effects introduced by symbol-rate sampling, it can be approximated as follows:
\begin{align}
L_{\rm m}({\bm\psi})
&\approx
\left\lceil
\frac{
(\tau_{\max}({\bm\psi})+T_{\rm p})
-
(\tau_{\min}({\bm\psi})-T_{\rm p})
}{
T_{\rm s}
}
\right\rceil
\nonumber\\
&=
\left\lceil
\frac{\Delta\tau({\bm\psi})+2T_{\rm p}}{T_{\rm s}}
\right\rceil
=
\left\lceil
B\left(\Delta\tau({\bm\psi})+2T_{\rm p}\right)
\right\rceil.
\label{channel_memory_approximation}
\end{align}
The term $B\Delta\tau({\bm\psi})$ represents the aperture-induced inter-path delay spread, whereas $2BT_{\rm p}$ represents the retained matched-pulse duration. To obtain a causal representation, define $\tilde g[q;{\bm\psi}]\triangleq\hat g[q+\ell_{\min}({\bm\psi});{\bm\psi}]$ for $q=0,\ldots,L_{\rm m}({\bm\psi})$. This re-indexing changes only the timing reference and does not alter the channel energy or its frequency-response magnitude. In the following, the dependence of $L_{\rm m}$, $\ell_{\min}$, and $\ell_{\max}$ on ${\bm\psi}$ is omitted when no ambiguity arises.

Finite-block transmission introduces block boundaries, so delayed components can cause inter-block interference (IBI). Its duration depends on both the \emph{inter-path delay spread} and the \emph{retained matched-pulse tails}, which are included in $L_{\rm m}$. We assume $N>L_{\rm m}$ and prepend a CP of length $L_{\rm cp}\geq L_{\rm m}$ to each data block. The IBI from the retained finite-memory channel then lies entirely within the CP and disappears after CP removal. The discarded pulse tails may cause residual IBI, which is negligible under the adopted truncation criterion.
\vspace{-10pt}
\begin{remark}
Equation \eqref{channel_memory_approximation} shows that a larger waveguide aperture can increase the required CP through $\Delta\tau({\bm\psi})$. A longer waveguide extends coverage by reducing the user-to-PA distances. In SWAN, it also activates more segments, which can provide additional array gain but may increase the differential delay spread and, consequently, the CP overhead.
\end{remark}
\vspace{-10pt}
Let ${\mathbf s}_N=[s[1],\ldots,s[N]]^{\mathsf T}$ denote the transmitted block. After CP removal, the retained samples satisfy $r_{\rm eff}[n]=\sum_{q=0}^{L_{\rm m}}\tilde g[q;{\bm\psi}]s[((n-1-q)\bmod N)+1]+w_{\rm eff}[n]$ for $n\in[N]$. Thus, ${\mathbf r}_{\rm eff}={\mathbf H}_{\rm eff}({\bm\psi}){\mathbf s}_N+{\mathbf w}_{\rm eff}$, where ${\mathbf H}_{\rm eff}({\bm\psi})\in\mathbb{C}^{N\times N}$
is the circulant matrix whose first column is
$[\tilde g[0;{\bm\psi}],\ldots,\tilde g[L_{\rm m};{\bm\psi}],
0,\ldots,0]^{\mathsf T}\in\mathbb{C}^{N\times1}$,
${\mathbf r}_{\rm eff}\triangleq
[r_{\rm eff}[1],\ldots,r_{\rm eff}[N]]^{\mathsf T}$, and
${\mathbf w}_{\rm eff}\triangleq
[w_{\rm eff}[1],\ldots,w_{\rm eff}[N]]^{\mathsf T}
\sim{\mathcal{CN}}({\mathbf 0},N_0{\mathbf I}_N)$ under the adopted
Nyquist pulse assumption.

Since ${\mathbf H}_{\rm eff}({\bm\psi})$ is circulant, it can be diagonalized by the unitary DFT matrix ${\mathbf F}_N$ as ${\mathbf H}_{\rm eff}({\bm\psi})={\mathbf F}_N^{\mathsf H}{\bm\Lambda}({\bm\psi}){\mathbf F}_N$, where
${\bm\Lambda}({\bm\psi})
={\mathsf{diag}}(\hat{\lambda}_N^{(1)}({\bm\psi}),\ldots,
\hat{\lambda}_N^{(N)}({\bm\psi}))$.
The $n$th eigenvalue is the $N$-point DFT of the causal channel
taps:
\begin{align}
\hat\lambda_N^{(n)}({\bm\psi})
=
\sum_{q=0}^{L_{\rm m}}
\tilde g[q;{\bm\psi}]
e^{-{\rm j}\frac{2\pi(n-1)q}{N}}=
e^{{\rm j}\frac{2\pi(n-1)\ell_{\min}}{N}}
\sum_{\ell=\ell_{\min}}^{\ell_{\max}}
\hat g[\ell;{\bm\psi}]
e^{-{\rm j}\frac{2\pi(n-1)\ell}{N}},
\quad n\in[N].
\nonumber
\end{align}
When the discarded channel tails are negligible, Nyquist-rate sampling with $T_{\rm s}=\frac{1}{B}$ gives
\begin{equation}
\hat\lambda_N^{(n)}({\bm\psi})
\approx
e^{{\rm j}\frac{2\pi(n-1)\ell_{\min}}{N}}
G(f_n;{\bm\psi}),
\label{eigenvalue_sc_fde}
\end{equation}
where
\begin{equation}
f_n
\triangleq
\begin{cases}
\frac{n-1}{NT_{\rm s}},
&
1\leq n\leq\lceil N/2\rceil,
\\
\frac{n-1-N}{NT_{\rm s}},
&
\lceil N/2\rceil<n\leq N.
\end{cases}
\label{frequency_bin_def}
\end{equation}

Applying the DFT to this block model converts the circular convolution into $N$ parallel frequency-bin channels:
\begin{equation}
\tilde{\mathbf r}_{\rm eff}={\mathbf F}_N{\mathbf r}_{\rm eff}={\bm\Lambda}({\bm\psi})\tilde{\mathbf s}_N+\tilde{\mathbf w}_{\rm eff}.
\nonumber
\end{equation}
Here, $\tilde{\mathbf s}_N={\mathbf F}_N{\mathbf s}_N$ and $\tilde{\mathbf w}_{\rm eff}={\mathbf F}_N{\mathbf w}_{\rm eff}$. Since ${\mathbf F}_N$ is unitary, $\tilde{\mathbf w}_{\rm eff}\sim{\mathcal{CN}}({\mathbf 0},N_0{\mathbf I}_N)$. The receiver applies a linear frequency-domain equalizer ${\mathbf W}\in\mathbb{C}^{N\times N}$ to $\tilde{\mathbf r}_{\rm eff}$ and transforms $\hat{\tilde{\mathbf s}}_N={\mathbf W}\tilde{\mathbf r}_{\rm eff}$ back to the time domain:
\begin{equation}
\hat{\mathbf s}_N={\mathbf F}_N^{\mathsf H}\hat{\tilde{\mathbf s}}_N={\mathbf G}_{\rm eq}{\mathbf s}_N+\hat{\mathbf w}_{\rm eff}.
\nonumber
\end{equation}
Here, ${\mathbf G}_{\rm eq}
\triangleq
{\mathbf F}_N^{\mathsf H}{\mathbf W}
{\bm\Lambda}({\bm\psi}){\mathbf F}_N$ denotes the overall time-domain response after equalization, and $\hat{\mathbf w}_{\rm eff}
\triangleq
{\mathbf F}_N^{\mathsf H}{\mathbf W}\tilde{\mathbf w}_{\rm eff}$ denotes the filtered noise. The entries of $\hat{\mathbf s}_N$ are subsequently processed by a symbolwise demapper. Since ${\mathbb E}[{\mathbf s}_N{\mathbf s}_N^{\mathsf H}]=P{\mathbf I}_N$, the unitarity of ${\mathbf F}_N$ gives ${\mathbb E}[\tilde{\mathbf s}_N\tilde{\mathbf s}_N^{\mathsf H}]=P{\mathbf I}_N$. The linear minimum mean-squared error (LMMSE) equalizer minimizes ${\mathbb E}[\|{\mathbf W}\tilde{\mathbf r}_{\rm eff}-\tilde{\mathbf s}_N\|^2]$. Because ${\bm\Lambda}({\bm\psi})$ is diagonal, this problem
decouples across the frequency bins. Hence,
${\mathbf W}={\mathsf{diag}}(W_1,\ldots,W_N)$, where $W_n
=
\frac{
P\hat\lambda_N^{(n)*}({\bm\psi})
}{
N_0+P|\hat\lambda_N^{(n)}({\bm\psi})|^2
}$ for $n\in[N]$. Thus, SC-FDE requires only one scalar equalization coefficient for each frequency bin. The $n$th recovered symbol can be written as follows:
\begin{equation}
[\hat{\mathbf s}_N]_n
=
[{\mathbf G}_{\rm eq}]_{n,n}s[n]
+
\sum_{n'\neq n}
[{\mathbf G}_{\rm eq}]_{n,n'}s[n']
+
[\hat{\mathbf w}_{\rm eff}]_n.
\label{single_user_scalar_equalized_model}
\end{equation}
The first term is the desired signal, whereas the second term is the residual inter-symbol interference (ISI) after linear equalization.

Define the SNR and LMMSE response of the $n$th frequency bin as $\tilde{\gamma}_n({\bm\psi})\triangleq\gamma_0|\hat\lambda_N^{(n)}({\bm\psi})|^2$ and $q_n({\bm\psi})\triangleq W_n\hat\lambda_N^{(n)}({\bm\psi})=\frac{\tilde{\gamma}_n({\bm\psi})}{1+\tilde{\gamma}_n({\bm\psi})}$, respectively. It follows that ${\mathbf G}_{\rm eq}={\mathbf F}_N^{\mathsf H}
{\mathsf{diag}}(q_1({\bm\psi}),\ldots,q_N({\bm\psi})){\mathbf F}_N$. Thus, ${\mathbf G}_{\rm eq}$ is circulant, and all its diagonal entries are identical:
$[{\mathbf G}_{\rm eq}]_{n,n}=\frac{1}{N}\sum_{i=1}^{N}q_i({\bm\psi})\triangleq\bar q({\bm\psi})$. Since $q_i({\bm\psi})$ is real and nonnegative, ${\mathbf G}_{\rm eq}{\mathbf G}_{\rm eq}^{\mathsf H}={\mathbf F}_N^{\mathsf H}{\mathsf{diag}}(q_1^2({\bm\psi}),\ldots,q_N^2({\bm\psi})){\mathbf F}_N$. Therefore,
\begin{equation}
\sum_{n'=1}^{N}
|[{\mathbf G}_{\rm eq}]_{n,n'}|^2
=
[{\mathbf G}_{\rm eq}{\mathbf G}_{\rm eq}^{\mathsf H}]_{n,n}
=
\frac{1}{N}
\sum_{i=1}^{N}q_i^2({\bm\psi}).
\nonumber
\end{equation}
Hence, the residual-ISI variance in \eqref{single_user_scalar_equalized_model} is given by
\begin{align}
\sigma_{\rm ISI}^2({\bm\psi})
=P\sum_{n'\neq n}|[{\mathbf G}_{\rm eq}]_{n,n'}|^2=P
\left[
\frac{1}{N}\sum_{i=1}^{N}q_i^2({\bm\psi})
-
\bar q^2({\bm\psi})
\right].\nonumber
\end{align}
The filtered-noise covariance is ${\mathbb E}[\hat{\mathbf w}_{\rm eff}\hat{\mathbf w}_{\rm eff}^{\mathsf H}]=N_0{\mathbf F}_N^{\mathsf H}{\mathbf W}{\mathbf W}^{\mathsf H}{\mathbf F}_N$. Its diagonal entries are identical and given by
\begin{equation}
\sigma_{\rm w}^2({\bm\psi})
=
\frac{N_0}{N}
\sum_{i=1}^{N}|W_i|^2
=
\frac{P}{N}
\sum_{i=1}^{N}
q_i({\bm\psi})
\left(1-q_i({\bm\psi})\right).
\nonumber
\end{equation}
Combining the residual ISI and filtered noise gives $\sigma_{\rm eff}^2({\bm\psi})=\sigma_{\rm ISI}^2({\bm\psi})+\sigma_{\rm w}^2({\bm\psi})=P\bar q({\bm\psi})(1-\bar q({\bm\psi}))$. The desired-signal power is $P\bar q^2({\bm\psi})$. Therefore, all recovered symbols have the same post-equalization signal-to-interference-plus-noise ratio (SINR), given by
\begin{equation}
\Gamma_{\rm FDE}({\bm\psi})
=\frac{P\bar q^2({\bm\psi})}{\sigma_{\rm eff}^2({\bm\psi})}
=\frac{\bar q({\bm\psi})}{1-\bar q({\bm\psi})}.
\nonumber
\end{equation}
For Gaussian signaling, the corresponding receiver-constrained rate can be expressed as follows:
\begin{equation}
\hat{\mathsf R}_{\rm sym}({\bm\psi})
=\log_2\left(1+\Gamma_{\rm FDE}({\bm\psi})\right)
=\log_2\left(\frac{1}{1-\bar q({\bm\psi})}\right).
\nonumber
\end{equation}
Since
\begin{equation}
1-\bar q({\bm\psi})
=
\frac{1}{N}
\sum_{n=1}^{N}
\left(
1-
\frac{\tilde{\gamma}_n({\bm\psi})}
{1+\tilde{\gamma}_n({\bm\psi})}
\right)
=
\frac{1}{N}
\sum_{n=1}^{N}
\frac{1}{1+\tilde{\gamma}_n({\bm\psi})},
\nonumber
\end{equation}
the rate becomes
\begin{equation}\nonumber
\hat{\mathsf R}_{\rm sym}({\bm\psi})
=
\log_2
\left(
\frac{1}{
\frac{1}{N}
\sum_{n=1}^{N}
\frac{1}{1+\tilde{\gamma}_n({\bm\psi})}
}
\right).
\end{equation}
Thus, the SC-FDE rate depends on the harmonic mean of $\{1+\tilde{\gamma}_n({\bm\psi})\}_{n=1}^{N}$. Frequency bins with small SNRs can therefore reduce the post-equalization performance. After inclusion of the CP overhead, the achievable spectral efficiency is
\begin{equation}\label{single_user_sc_fde_achievable_rate_pre}
\hat{\mathsf R}({\bm\psi})
=
\frac{N}{N+L_{\rm cp}}
\log_2
\left(
\frac{1}{
\frac{1}{N}
\sum_{n=1}^{N}
\frac{1}{1+\tilde{\gamma}_n({\bm\psi})}
}
\right).
\end{equation}
Using \eqref{eigenvalue_sc_fde}, it can be approximated as follows:
\begin{equation}
\hat{\mathsf R}({\bm\psi})
\approx
\frac{N}{N+L_{\rm cp}}
\log_2
\left(
\frac{N}{
\sum_{n=1}^{N}
\frac{1}{
1+\gamma_0|G(f_n;{\bm\psi})|^2
}
}
\right).\label{single_user_sc_fde_rate_frequency}
\end{equation}
Equation \eqref{single_user_sc_fde_rate_frequency} can replace the infinite-blocklength objective in the PA-placement design.

For fixed $L_{\rm cp}$ and $L_{\rm m}$, as $N\rightarrow\infty$,
the CP overhead vanishes and the frequency grid converges to the
continuous baseband interval. Hence,
\begin{equation}
\lim_{N\rightarrow\infty}
\hat{\mathsf R}({\bm\psi})
=
\log_2
\left[
\left(
\frac{1}{B}
\int_{-\frac{B}{2}}^{\frac{B}{2}}
\frac{
{\rm d}f
}{
1+\gamma_0|G(f;{\bm\psi})|^2
}
\right)^{-1}
\right].
\end{equation}
This limit is the asymptotic receiver-constrained rate of linear
SC-FDE, not the unconstrained Shannon rate. Jensen's inequality gives
\begin{align}
\lim_{N\rightarrow\infty}\hat{\mathsf R}({\bm\psi})
\leq
\frac{1}{B}
\int_{-\frac{B}{2}}^{\frac{B}{2}}
\log_2
(
1+\gamma_0|G(f;{\bm\psi})|^2
)
{\rm d}f
={\mathsf R}({\bm\psi}).\nonumber
\end{align}
Thus, linear SC-FDE generally incurs a rate loss relative to the infinite-blocklength benchmark. Symbolwise linear equalization reduces the FS channel to a common post-equalization SINR and does not exploit the frequency bins through optimal joint decoding.
\subsection{Multiuser Case}
We now extend the finite-block CP-SC-FDE model to the multiuser uplink case. We assume that uplink synchronization, such as timing advance, compensates for the user-dependent bulk arrival offsets. For user $k$, let ${\mathbf s}_{k,N}$ denote the transmitted block and let ${\mathbf H}_{k,{\rm eff}}({\bm\psi})$ denote the corresponding circulant channel matrix after CP removal. The received block is given by
\begin{equation}
{\mathbf r}_{\rm eff}
=
\sum_{k=1}^{K}
{\mathbf H}_{k,{\rm eff}}({\bm\psi}){\mathbf s}_{k,N}
+
{\mathbf w}_{\rm eff}.
\label{linear_multiuser_system_model_sc_fde}
\end{equation}
To obtain a valid circular-convolution model for all users, the common CP length should satisfy $L_{\rm cp}\geq\max_{k\in[K]}L_{k,{\rm m}}$, where $L_{k,{\rm m}}$ denotes the retained channel memory of user $k$ after truncation and causal re-indexing. Each circulant channel matrix can be diagonalized as ${\mathbf H}_{k,{\rm eff}}({\bm\psi})={\mathbf F}_N^{\mathsf H}{\bm\Lambda}_k({\bm\psi}){\mathbf F}_N$, where ${\bm\Lambda}_k({\bm\psi})={\mathsf{diag}}(\hat{\lambda}_{k,N}^{(1)}({\bm\psi}),\ldots,\hat{\lambda}_{k,N}^{(N)}({\bm\psi}))$. Applying the DFT to \eqref{linear_multiuser_system_model_sc_fde} gives $\tilde{\mathbf r}_{\rm eff}=\sum_{k=1}^{K}{\bm\Lambda}_k({\bm\psi})\tilde{\mathbf s}_{k,N}+\tilde{\mathbf w}_{\rm eff}$, where $\tilde{\mathbf s}_{k,N}={\mathbf F}_N{\mathbf s}_{k,N}$ and
$\tilde{\mathbf w}_{\rm eff}
={\mathbf F}_N{\mathbf w}_{\rm eff}
\sim{\mathcal{CN}}({\mathbf 0},N_0{\mathbf I}_N)$.

We consider parallel decoding, where each user is decoded
independently and the signals from the remaining users are treated
as Gaussian interference. When decoding user $k$, the
interference-plus-noise variance at frequency bin $n$ is
\begin{equation}
\sigma_{k,n}^2({\bm\psi})
=
N_0+
\sum_{k'\neq k}
P_{k'}
|\hat{\lambda}_{k',N}^{(n)}({\bm\psi})|^2.
\nonumber
\end{equation}
The corresponding frequency-bin SINR is
\begin{align}
\gamma_{k,n}^{\rm P}({\bm\psi})
\triangleq\frac{P_k|\hat{\lambda}_{k,N}^{(n)}({\bm\psi})|^2}{\sigma_{k,n}^2({\bm\psi})}=
\frac{P_k|\hat{\lambda}_{k,N}^{(n)}({\bm\psi})|^2}
{N_0+\sum_{k'\neq k}P_{k'}|\hat{\lambda}_{k',N}^{(n)}({\bm\psi})|^2}\approx\frac{
A_k(f_n;{\bm\psi})
}{
1+\hat A_k(f_n;{\bm\psi})
},
\end{align}
where $A_{k}(f;{\bm\psi})=\frac{P_k}{N_0}|G_k(f;{\bm\psi})|^2$, $\hat A_{k}(f;{\bm\psi})=\sum_{k'\neq k}A_{k'}(f;{\bm\psi})$, and $f_n$ is defined in \eqref{frequency_bin_def}.

The single-user CP-SC-FDE derivation that leads to \eqref{single_user_sc_fde_achievable_rate_pre} gives the receiver-constrained sum-rate under Gaussian signaling:
\begin{equation}
\hat{\mathsf R}_{\rm P}({\bm\psi})
=
\frac{N}{N+L_{\rm cp}}
\sum_{k=1}^{K}
\log_2
\left(
\frac{1}{
\frac{1}{N}\sum_{n=1}^{N}
\frac{1}{
1+\gamma_{k,n}^{\rm P}({\bm\psi})
}
}
\right).
\label{multiple_user_sc_fde_achievable_rate}
\end{equation}
Thus, the rate of each user depends on the harmonic mean of its
frequency-bin SINR terms. Consequently, frequency bins with strong
multiuser interference can substantially reduce the corresponding
post-equalization rate.

For fixed $L_{\rm cp}$ and $\{L_{k,{\rm m}}\}_{k=1}^{K}$, as
$N\rightarrow\infty$, the CP overhead vanishes and the discrete
frequency grid converges to the continuous baseband interval.
Using the frequency-response approximation in
\eqref{eigenvalue_sc_fde}, we obtain
\begin{equation}
\lim_{N\rightarrow\infty}
\hat{\mathsf R}_{\rm P}({\bm\psi})
=
\sum_{k=1}^{K}
\log_2
\left[
\left(
\frac{1}{B}
\int_{-\frac{B}{2}}^{\frac{B}{2}}
\frac{
{\rm d}f
}{
1+
\frac{
A_k(f;{\bm\psi})
}{
1+\hat A_k(f;{\bm\psi})
}
}
\right)^{-1}
\right].
\end{equation}
By Jensen's inequality, this asymptotic CP-SC-FDE rate satisfies
\begin{equation}
\lim_{N\rightarrow\infty}
\hat{\mathsf R}_{\rm P}({\bm\psi})
\leq
{\mathsf R}_{\rm P}({\bm\psi}).
\nonumber
\end{equation}
Therefore, parallel CP-SC-FDE incurs a rate loss relative to ideal frequency-domain coding because each user is reduced to a common post-equalization SINR and the remaining users are treated as interference.

The rate in \eqref{multiple_user_sc_fde_achievable_rate} remains an explicit function of the PA position vector ${\bm\psi}$, which determines both the desired FS gains and the spectral interference among users. It can therefore replace the infinite-blocklength parallel-decoding objective in the multiuser PA-placement problem, and the same element-wise optimization framework can be applied.

\section{Numerical Results}\label{Section:Numerical Results}
We validate the analytical results and the proposed PA-placement design through numerical simulations. Unless stated otherwise, $f_{\rm c}=28$ GHz, $n_{\rm eff}=1.4$, $d=3$ m, $D_y=20$ m, $L=1$ m, $\kappa=0.08$ dB/m, $\Delta=\lambda/2$, $B=50$ MHz, $I_{\rm GL}=50$, and $\gamma_0=70$ dB for $k\in[K]$. The waveguide covers $[-D_x/2,D_x/2]$ and contains $M=D_x/L$ segments. For the average-rate results, the users are independently and uniformly distributed over $[-D_x/2,D_x/2]\times[-D_y/2,D_y/2]$, and each simulation point represents an average over $1000$ deployments. We set $K=4$ in the multiuser simulations. The ideal sinc pulse is used with $T_{\rm s}=1/B$. Each one-dimensional PA update searches a uniform grid of $Q=1001$ points. An FS design optimizes the placement under the exact FS model, whereas an FF design uses the conventional FF approximation. An FS evaluation always evaluates the resulting placement with the exact channel. An FF evaluation instead reports the rate predicted by the FF model.
\subsection{Frequency-Selectivity Characterization}

\begin{figure}[!t]
\centering
\includegraphics[width=0.8\textwidth]{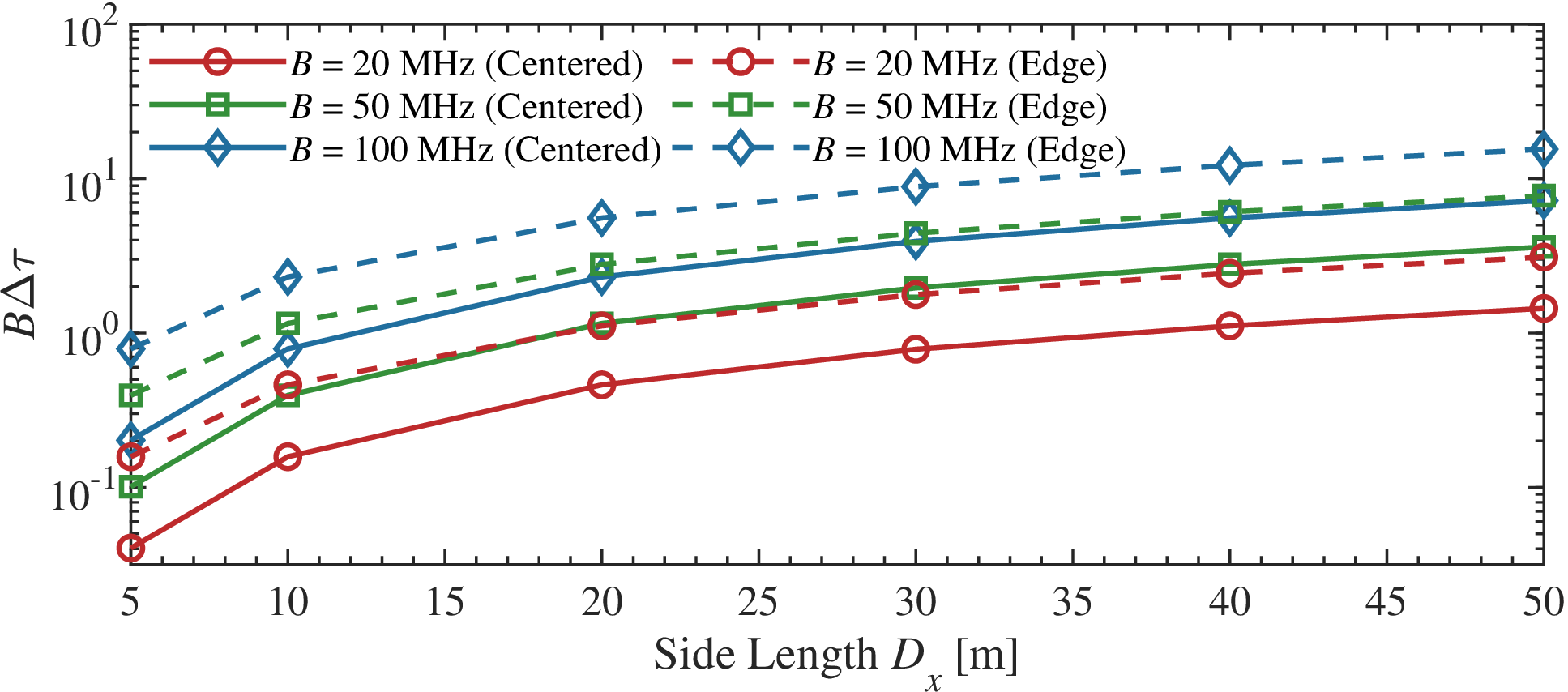}
\caption{Bandwidth-delay-spread product versus $D_x$.}
\label{Figure_Single_User_Figure1}
\vspace{-10pt}
\end{figure}

{\figurename} {\ref{Figure_Single_User_Figure1}} plots the bandwidth-delay-spread product $B\Delta\tau$ for midpoint PA placement, i.e., $\psi_{m}=\psi_{0}^{m}+L/2$ for $m\in[M]$. The centered and edge users lie directly below the waveguide axis. Although the considered bandwidths are moderate, $B\Delta\tau$ increases rapidly with $D_x$ because a wider aperture increases the disparity among the PA-assisted path delays. The edge user has a larger value than the centered user because its propagation distances are more asymmetric. In particular, $B\Delta\tau$ exceeds $1$ for several ordinary aperture-bandwidth combinations. Thus, a small fractional bandwidth alone does not ensure an FF PASS channel; the physical aperture must also be considered. 

\begin{figure}[!t]
\centering
    \subfigure[$D_x=5$ m.]
    {
        \includegraphics[height=0.35\textwidth]{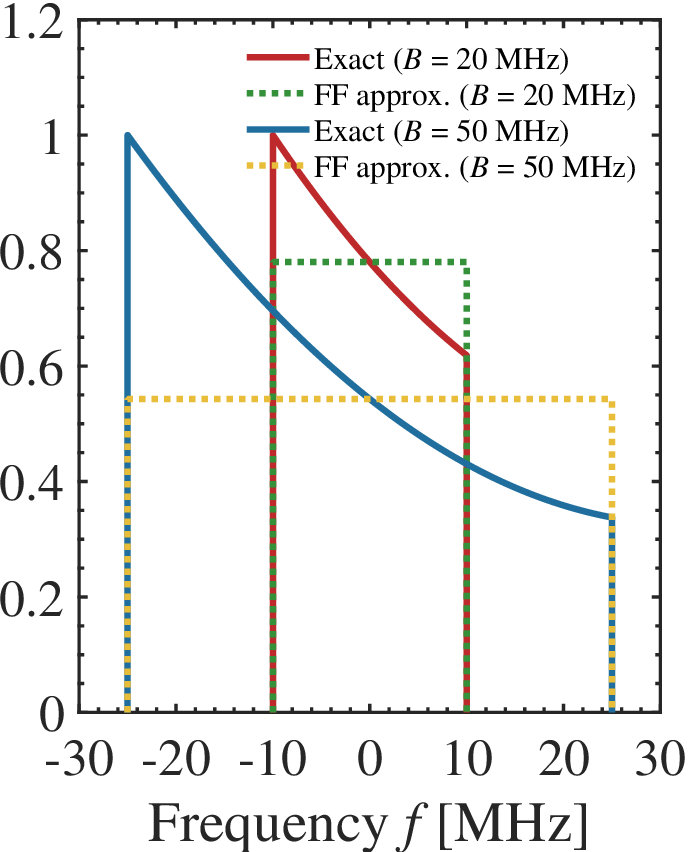}
        \label{Figure: Single_User_Figure2a}
    }
    \subfigure[$D_x=20$ m.]
    {
        \includegraphics[height=0.35\textwidth]{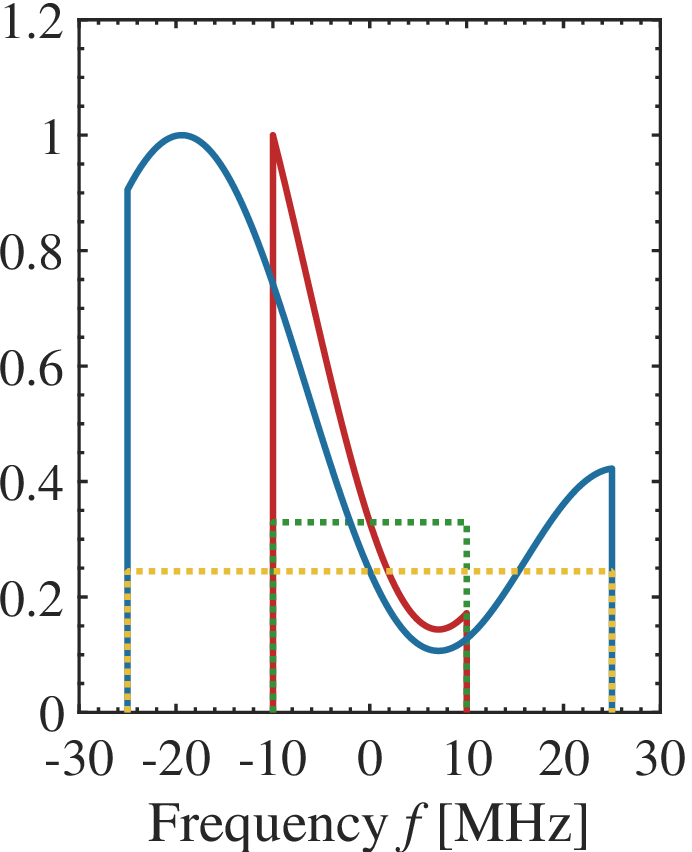}
        \label{Figure: Single_User_Figure2b}
    }
    \subfigure[$D_x=40$ m.]
    {
        \includegraphics[height=0.35\textwidth]{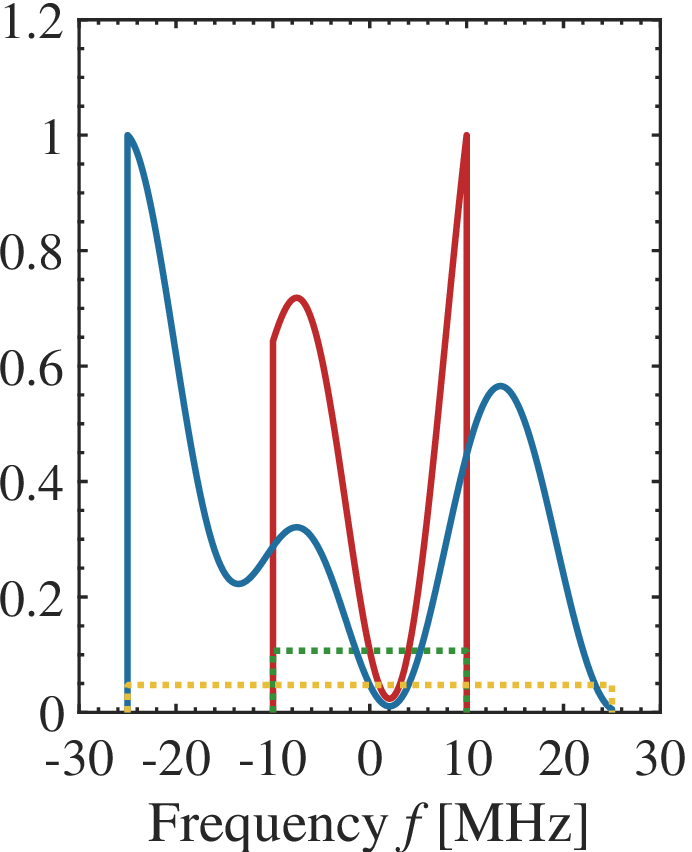}
        \label{Figure: Single_User_Figure2c}
    }
\caption{Normalized channel power spectrum.}
\label{Figure: Single_User_Figure2}
\vspace{-10pt}
\end{figure}

{\figurename} {\ref{Figure: Single_User_Figure2}} illustrates the spectral distortion of a centered user for midpoint placement. For each bandwidth, the plotted exact spectrum is $\frac{|G(f;{\bm\psi})|^2}{\max_{|\nu|\leq B/2}|G(\nu;{\bm\psi})|^2}$. The FF curve uses the constant value $|h_{\rm eff}({\bm\psi})|^2$ with the same normalization. At $D_x=5$ m, the exact response changes relatively smoothly over each band. As $D_x$ increases to $20$ m and $40$ m, pronounced frequency variations and deep spectral notches appear. The constant FF response cannot capture these features. This observation is consistent with the increase in $B\Delta\tau$ shown in {\figurename} {\ref{Figure_Single_User_Figure1}}. These two figures together highlight the necessity of considering frequency selectivity in practical PASS channels.

\begin{figure}[!t]
\centering
\subfigure[Centered user.]
    {
        \includegraphics[height=0.35\textwidth]{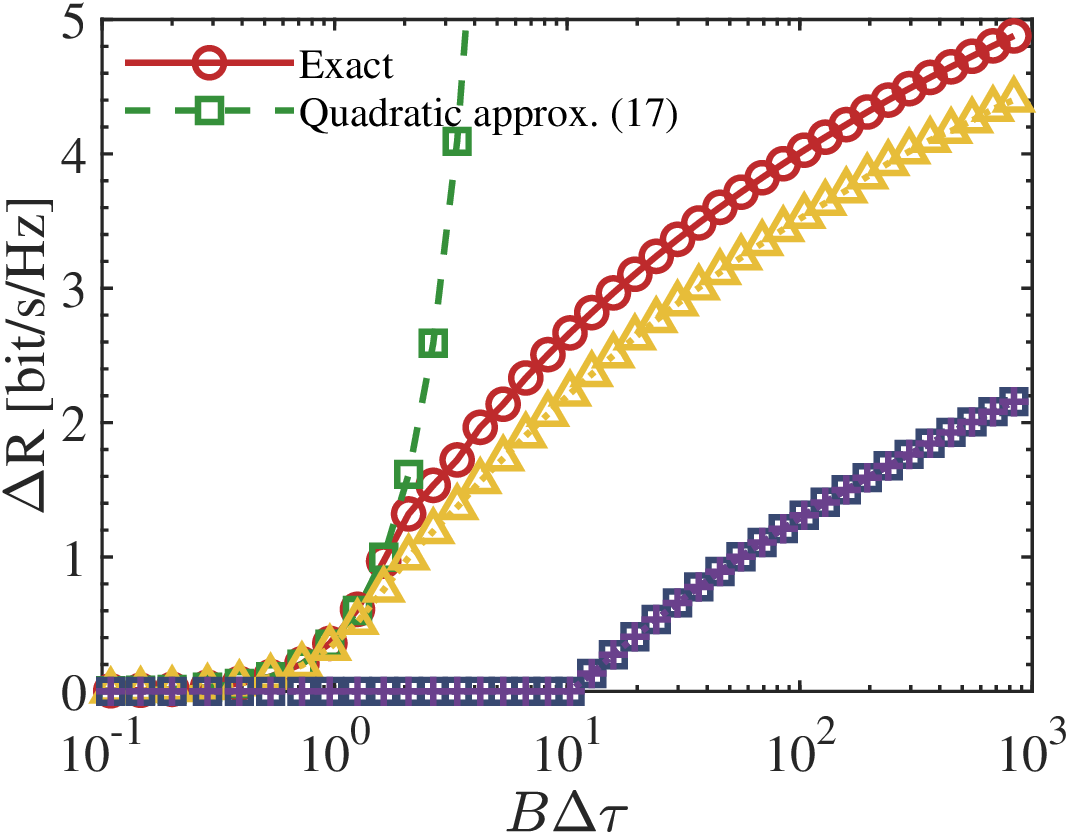}
        \label{Figure: Single_User_Figure3a}
    }
    \subfigure[Edge user.]
    {
        \includegraphics[height=0.35\textwidth]{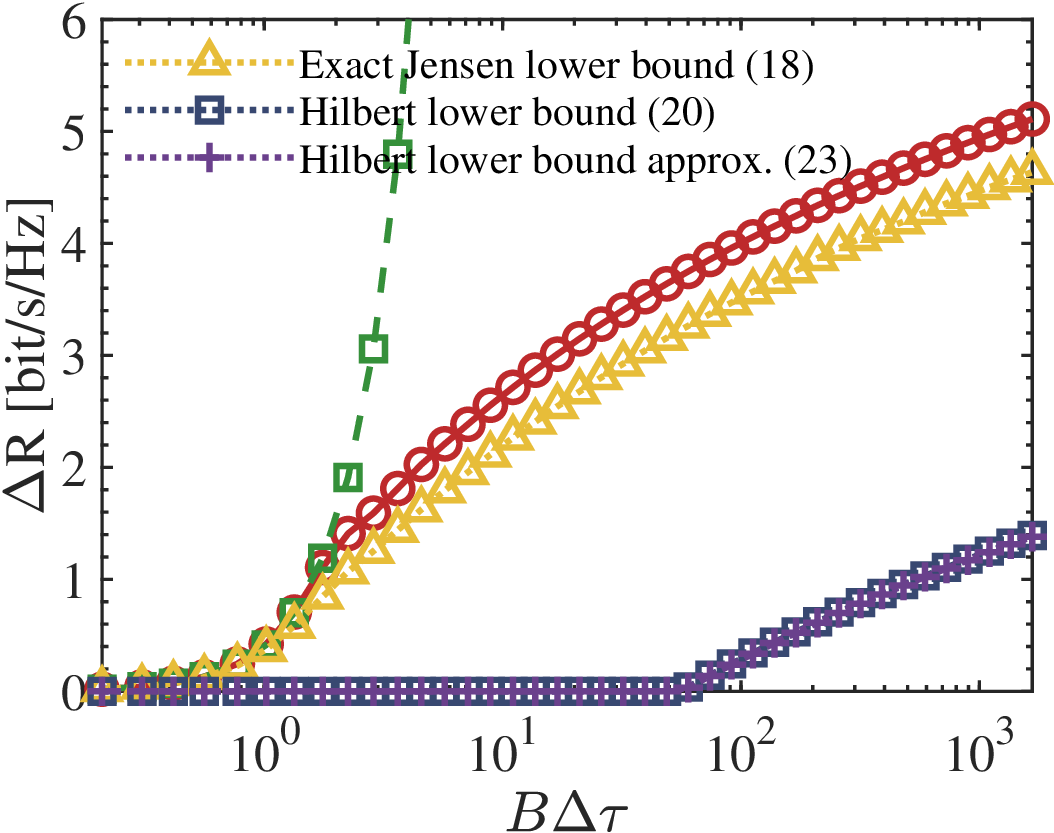}
        \label{Figure: Single_User_Figure3b}
    }
\caption{FF rate overestimation $\Delta{\mathsf{R}}$ versus $B\Delta\tau$ for $L=0.1$ m.}
\label{Figure: Single_User_Figure3}
\vspace{-10pt}
\end{figure}

{\figurename} {\ref{Figure: Single_User_Figure3}} validates the small- and large-delay-spread characterizations of the FF rate overestimation for $L=0.1$ m under carrier-phase-aligned PA placement \cite{ouyang2025array,ouyang2025uplink}. The horizontal axis shows the resulting bandwidth-delay-spread product $B\Delta\tau$, obtained by varying the waveguide aperture $D_x$ at a fixed bandwidth $B$. For small $B\Delta\tau$, the quadratic approximation in \eqref{rate_gap_quadratic_onset} closely follows the exact FF rate overestimation, which verifies the small-delay-spread characterization in Theorem \ref{them:rate_difference}. As $B\Delta\tau$ increases, the approximation gradually deviates from the exact result because it is valid only in the small-delay-spread regime. The Jensen lower bound in \eqref{rate_gap_Jensen_bound} captures the growth of the exact overestimation over the entire plotted range. The Hilbert lower bound is obtained from \eqref{rate_gap_effective_paths}, whereas its integral approximation substitutes \eqref{effective_paths_integral_approx} and \eqref{effective_snr_integral_approx} into \eqref{rate_gap_effective_paths_case}. The two curves nearly overlap, which validates the sum-to-integral approximation. The Hilbert bound remains zero until $M_{\rm eff}>1+f_{\rm c}/B$, after which it increases with the delay spread and approaches the double-logarithmic behavior characterized in \eqref{rate_gap_aperture_scaling_tau}. These results therefore also validate the large-delay-spread characterization.

\subsection{Infinite-Blocklength Rates}

\begin{figure}[!t]
\centering
    \subfigure[$B=20$ MHz.]
    {
        \includegraphics[height=0.3\textwidth]{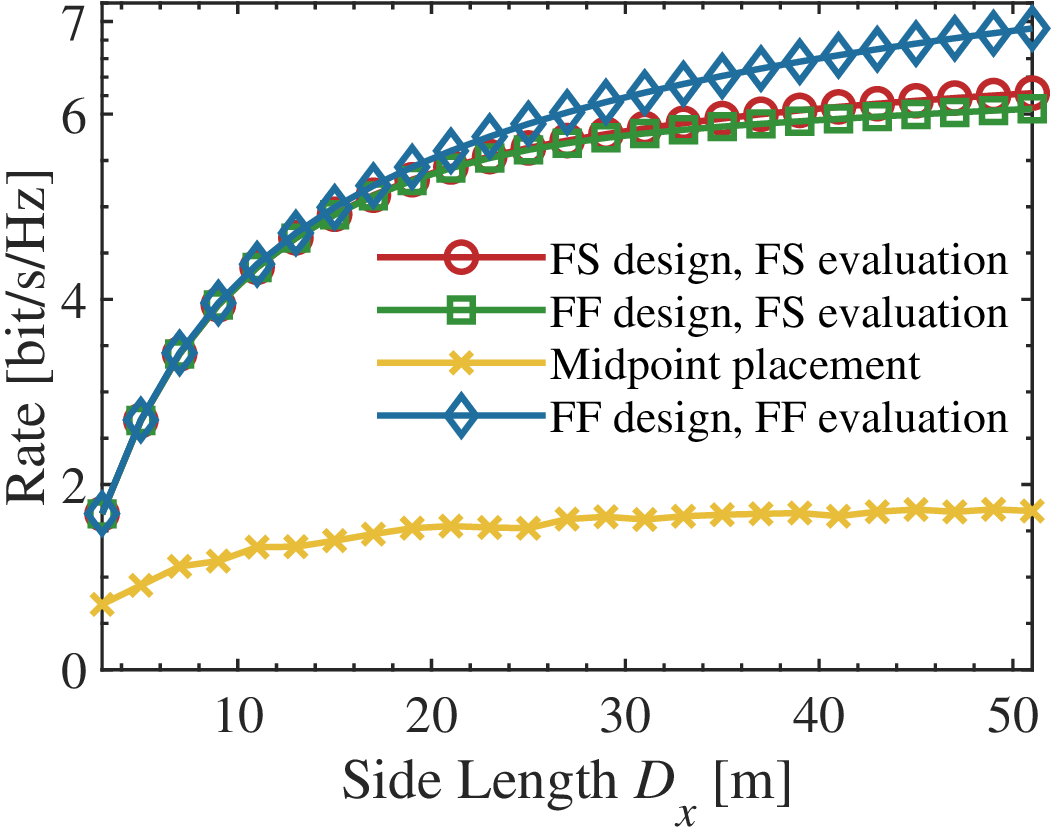}
        \label{Figure: Infinite_Blocklength_Figure1a}
    }
    \subfigure[$B=50$ MHz.]
    {
        \includegraphics[height=0.3\textwidth]{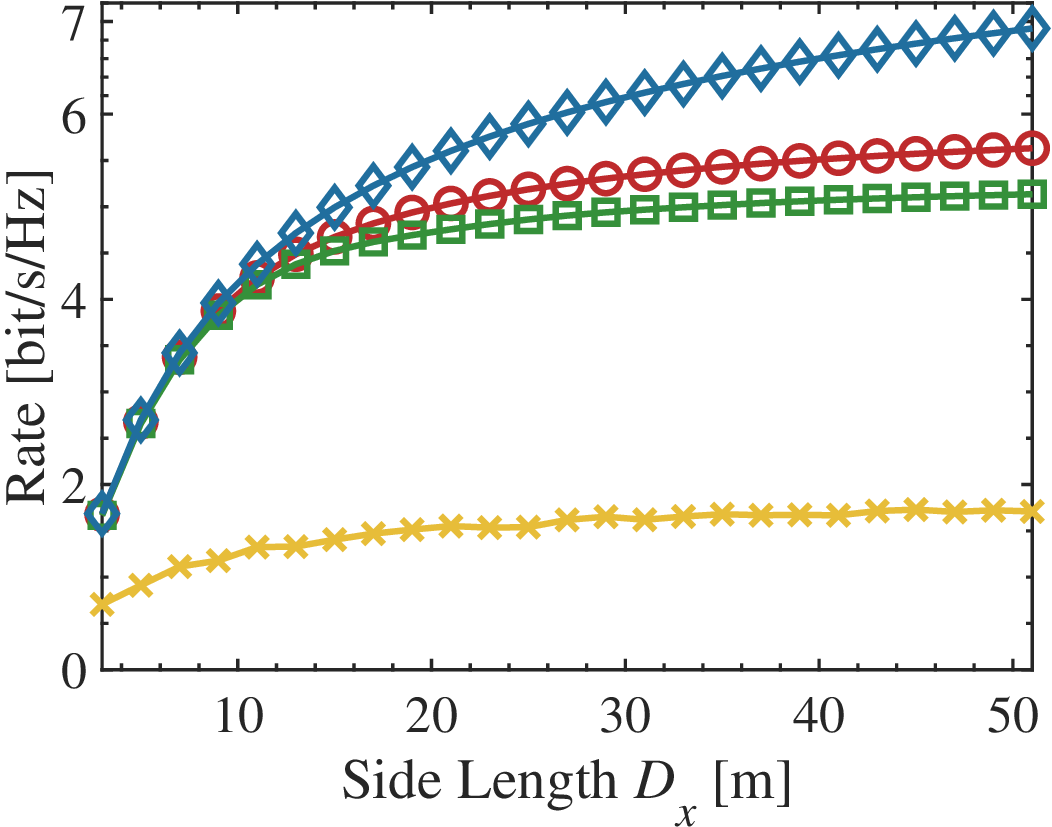}
        \label{Figure: Infinite_Blocklength_Figure1b}
    }\\
    \subfigure[$B=100$ MHz.]
    {
        \includegraphics[height=0.3\textwidth]{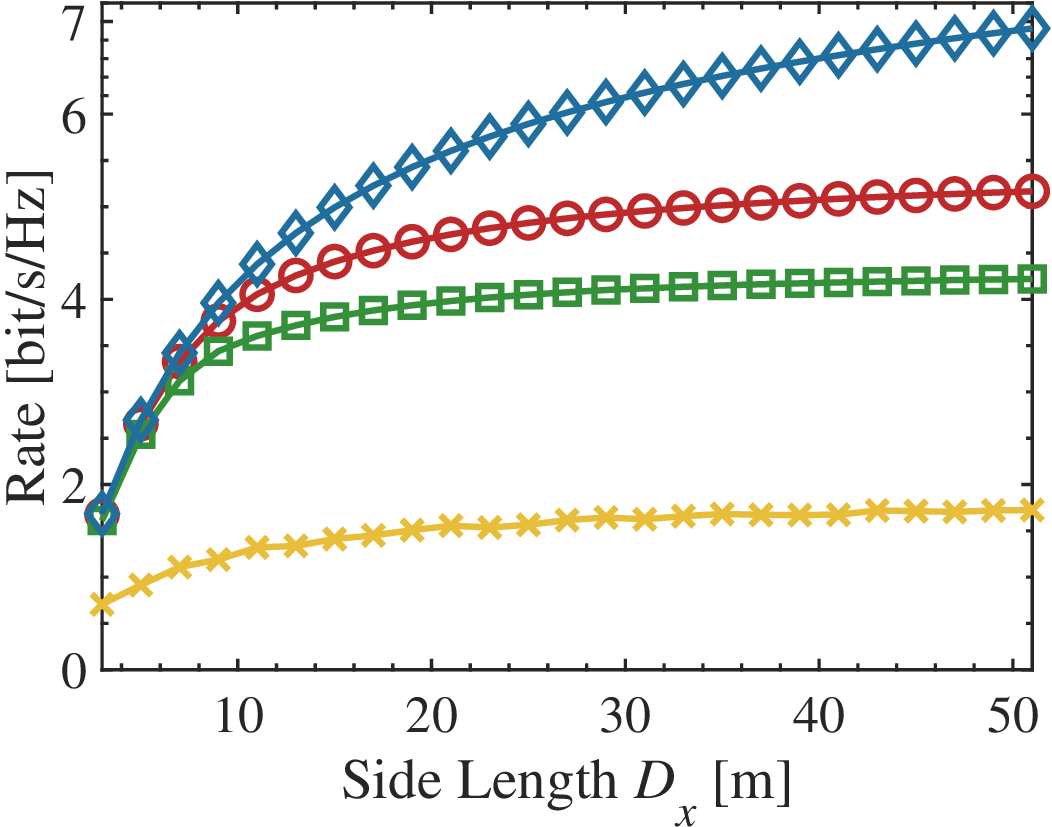}
        \label{Figure: Infinite_Blocklength_Figure1c}
    }
    \subfigure[$B=200$ MHz.]
    {
        \includegraphics[height=0.3\textwidth]{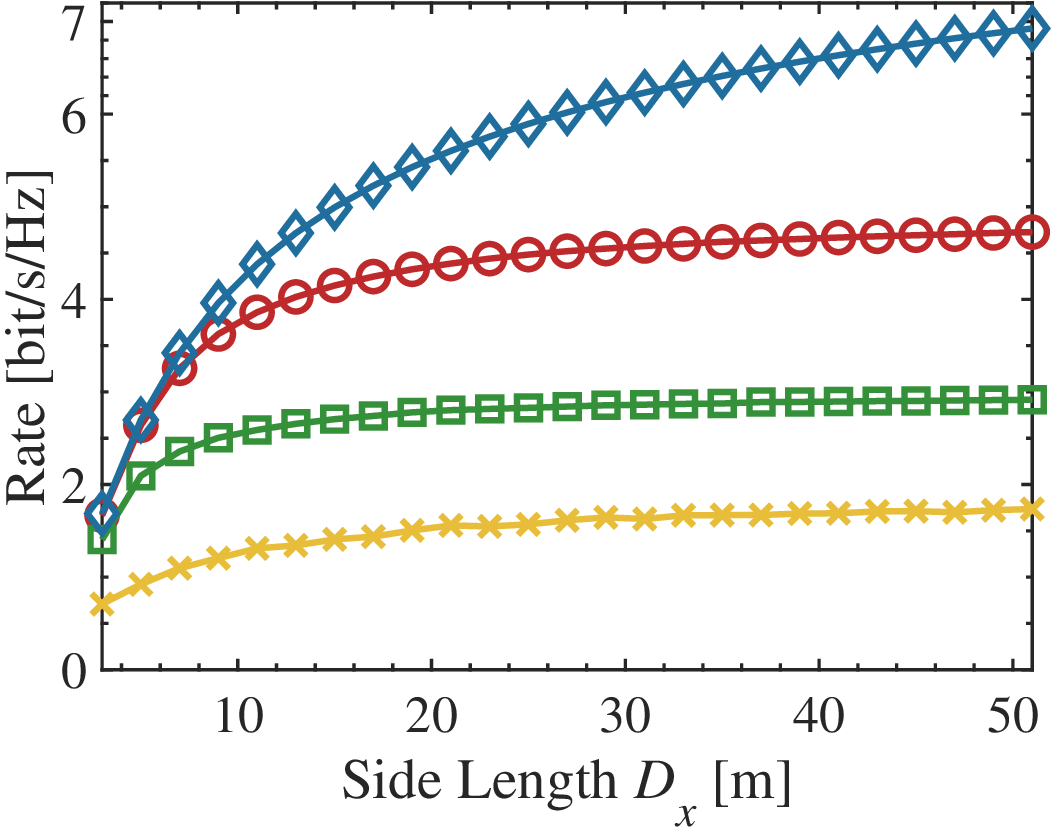}
        \label{Figure: Infinite_Blocklength_Figure1d}
    }
\caption{Infinite-blocklength rate versus $D_x$.}
\label{Figure: Infinite_Blocklength_Figure}
\vspace{-10pt}
\end{figure}

{\figurename} {\ref{Figure: Infinite_Blocklength_Figure}} compares the single-user rates for four bandwidths. The midpoint curve evaluates fixed midpoint placement with the FS model. At small $B$ and $D_x$, the FS- and FF-designed placements achieve similar exact rates because the channel is nearly FF. Their separation grows with either parameter, and hence with $B\Delta\tau$. The FS design consistently provides the largest exact rate, which demonstrates the effectiveness of the element-wise optimization under the FS objective. In contrast, the FF evaluation becomes increasingly optimistic as the bandwidth or aperture grows. Midpoint placement performs poorly because it accounts for neither carrier-phase alignment nor FS distortion. A larger $D_x$ activates more segments and generally improves the FS-aware rate through the resulting array gain, although frequency selectivity reduces the attainable gain.

\begin{figure}[!t]
\centering
\includegraphics[width=0.8\textwidth]{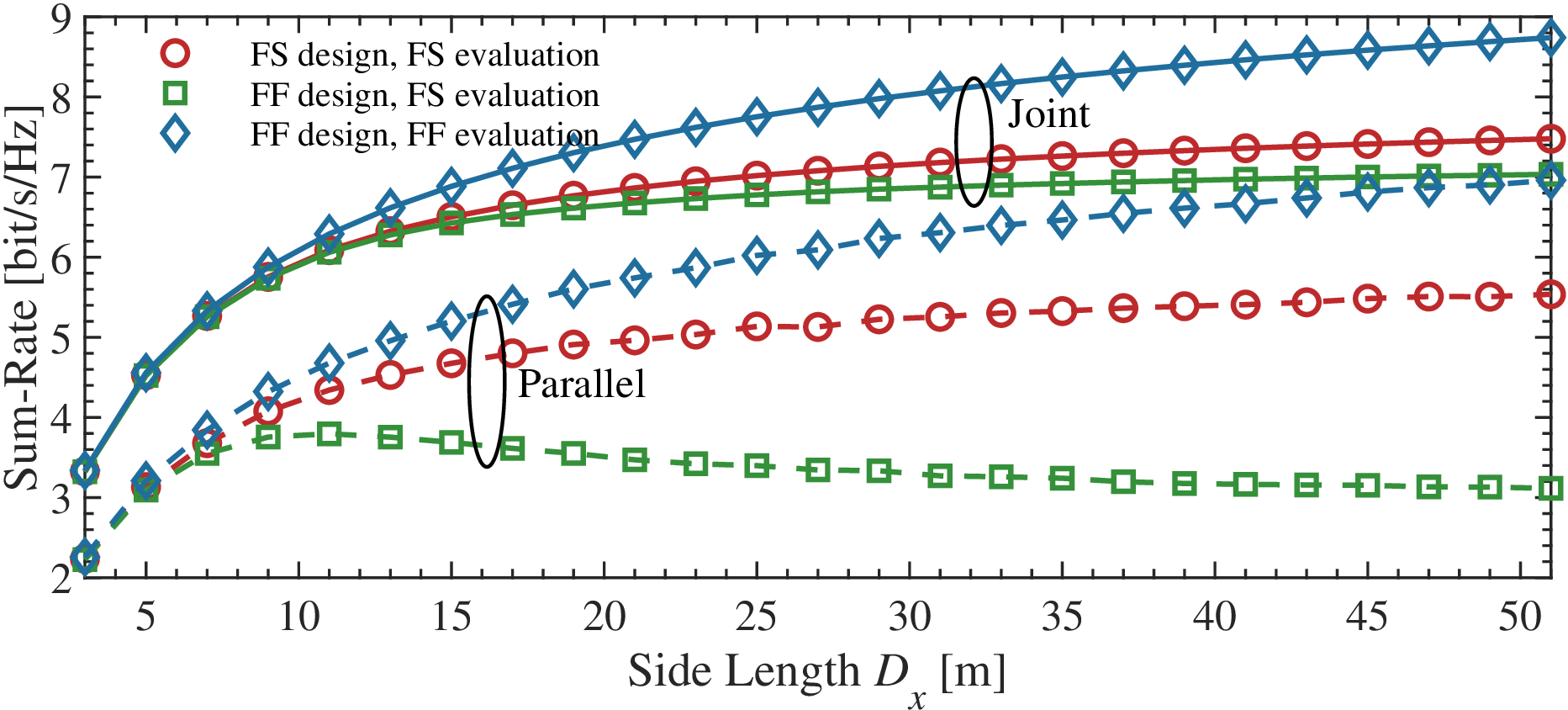}
\caption{Infinite-blocklength sum-rate versus $D_x$.}
\label{Infinite_Blocklength_Figure2}
\vspace{-10pt}
\end{figure}

{\figurename} {\ref{Infinite_Blocklength_Figure2}} presents the corresponding multiuser sum-rates at $B=50$ MHz. Joint decoding outperforms parallel decoding because it avoids treating the other users as interference. For both decoding strategies, the FS design gives the largest sum-rate under exact evaluation. The benefit is especially pronounced for parallel decoding because PA placement must manage both the desired spectra and the frequency-dependent inter-user interference. Under the FF design, the exact parallel-decoding sum-rate first increases and then decreases with $D_x$. The FF objective favors carrier-frequency alignment but does not represent the spectral interference seen by the exact receiver. The resulting mismatch can outweigh the additional array gain. In contrast, the FS design preserves a monotonic improvement over the simulated range. These results also show why the placement must be optimized separately for joint and parallel decoding.

\subsection{Finite-Block CP-SC-FDE Rates}
For finite-block transmission, we consider $11$ block sizes over $256\leq N\leq8192$ and set $\epsilon_{\rm p}=10^{-3}$. The sinc matched pulse then satisfies $BT_{\rm p}=112$, or $T_{\rm p}=2.24~\mu$s at $B=50$ MHz, and retains at least $99.9\%$ of its energy. For each realization, the CP length is set to the retained channel memory for one user and to the largest retained user-channel memory for multiple users. Thus, $L_{\rm cp}=L_{\rm m}$ in the single-user case and $L_{\rm cp}=\max_kL_{k,{\rm m}}$ in the multiuser case. 

\begin{figure}[!t]
\centering
    \subfigure[Single-user case.]
    {
        \includegraphics[width=0.8\textwidth]{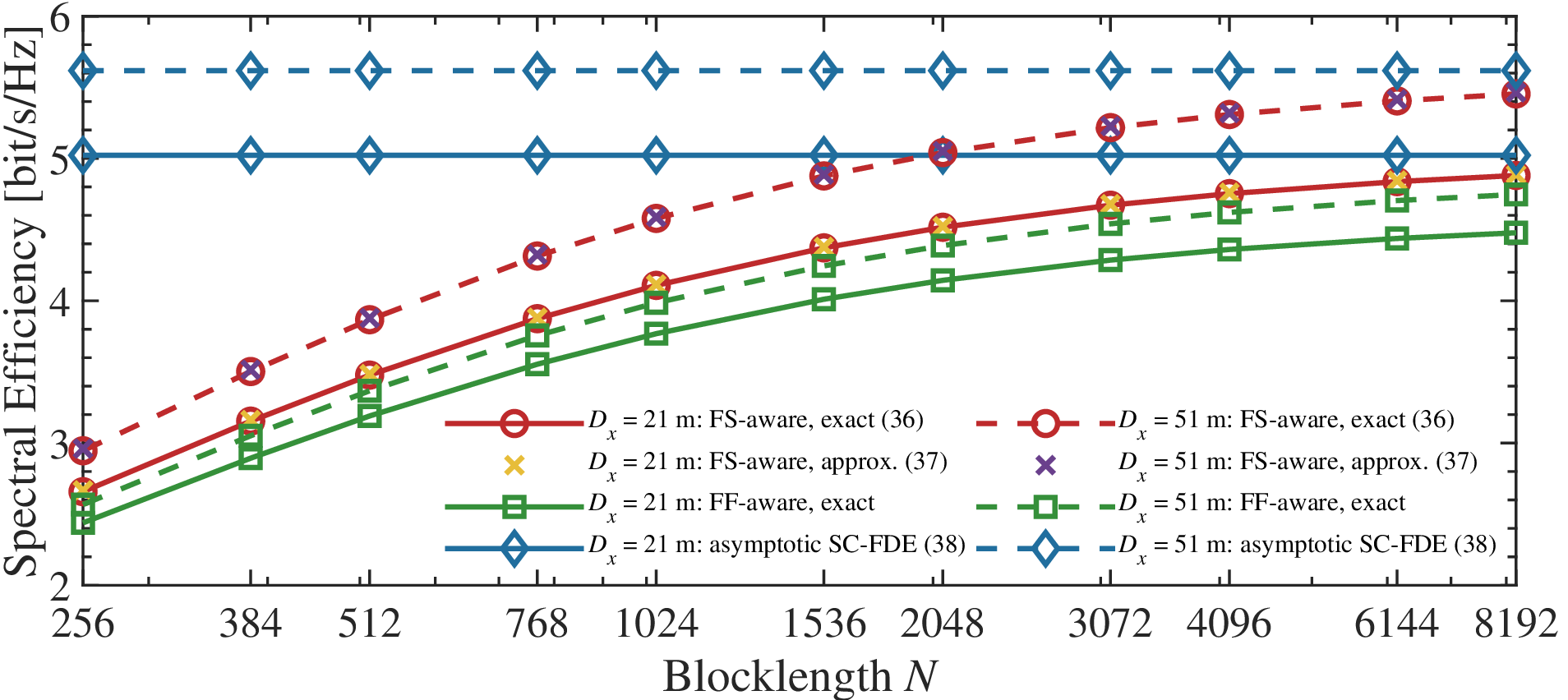}
        \label{Figure: Finite_Blocklength_Figure1}
    }\\
    \subfigure[Multiuser case.]
    {
        \includegraphics[width=0.8\textwidth]{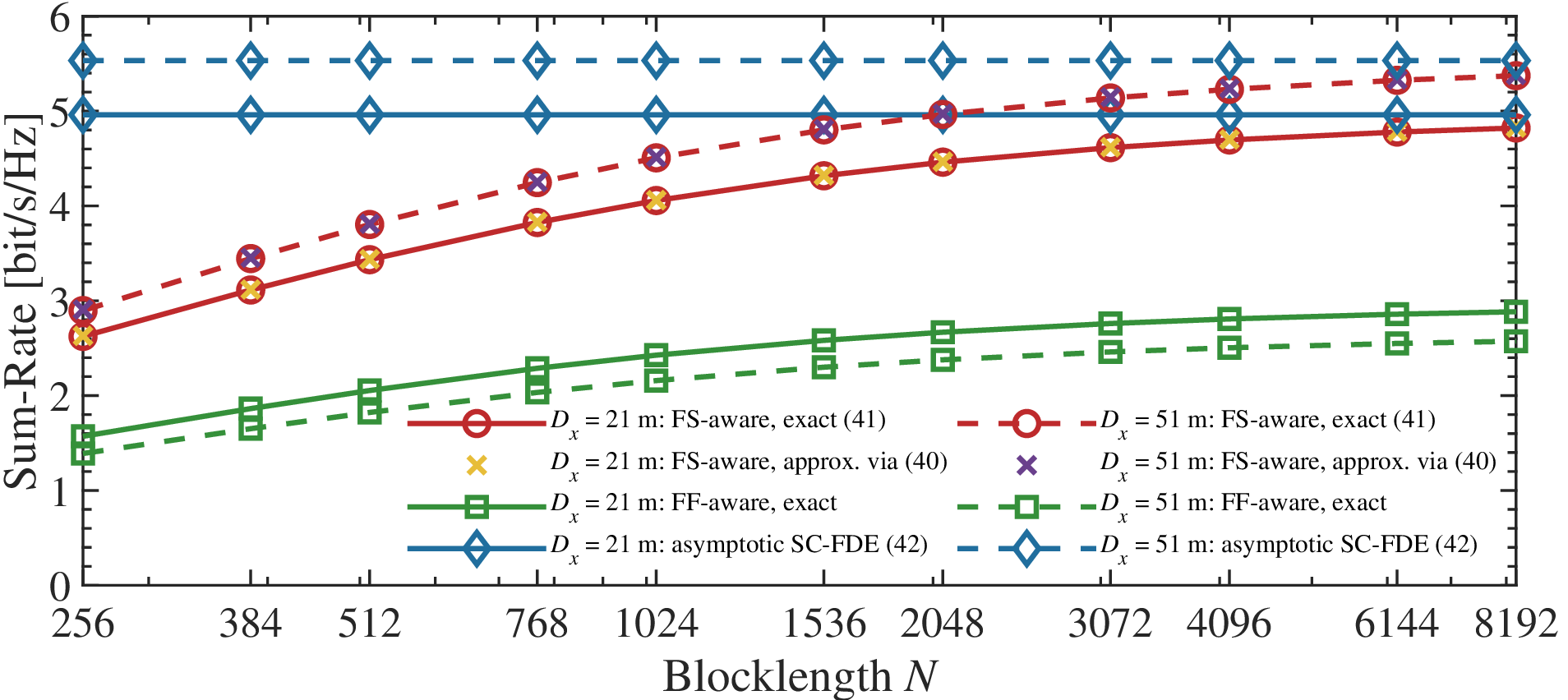}
        \label{Figure: Finite_Blocklength_Figure2}
    }
\caption{Finite-block CP-SC-FDE rate versus blocklength $N$.}
\label{Figure: Finite_Blocklength_Figure}
\vspace{-10pt}
\end{figure}

{\figurename} {\ref{Figure: Finite_Blocklength_Figure}} shows the CP-SC-FDE results for $D_x\in\{21,51\}$ m. The exact and frequency-response-based approximate evaluations of the FS-aware placement nearly coincide, which supports \eqref{single_user_sc_fde_rate_frequency} and its multiuser counterpart. As $N$ increases, the CP overhead decreases and the rates approach the asymptotic SC-FDE limits derived in Section \ref{Section: Finite-Block CP-SC-FDE}. The FS-aware placement substantially outperforms the FF placement for both apertures. In the single-user case, the larger aperture provides the higher rate because its additional segments yield more array gain. In the multiuser case, however, the FF result for $D_x=51$ m is lower than that for $D_x=21$ m. This reversal agrees with the infinite-blocklength parallel-decoding result in {\figurename} {\ref{Infinite_Blocklength_Figure2}}: an FF placement cannot control the exact frequency-dependent inter-user interference. Once the FS-aware objective is used, $D_x=51$ m again outperforms $D_x=21$ m. The comparison highlights the importance of incorporating frequency selectivity into PA placement for finite-block multiuser transmission.

\section{Conclusion}\label{Section: Conclusion}
We showed that extending PASS deployment reduces user-to-PA path loss and improves coverage but also enlarges the aperture-induced differential propagation delays. Consequently, aperture-induced frequency selectivity can become significant even at moderate bandwidths. In the infinite-blocklength regime, we proved that the FF model overestimates the exact FS rate under carrier-phase-aligned PA placement. For small delay spread $\Delta\tau$, this overestimation grows approximately quadratically with $\Delta\tau$. As the waveguide aperture and the resulting delay spread become large, it grows at least double logarithmically with $\Delta\tau$. For finite-block transmission, we characterized the achievable rate of PASS under a practical CP-SC-FDE framework and showed that the required CP must cover both the path-delay spread and the retained pulse tails. Its length therefore increases with the delay spread and, consequently, with the PASS aperture. For all considered cases, we proposed FS-aware PA-placement methods. The results showed that FF placement can forfeit or even reverse the gain from additional segments, whereas FS-aware placement preserves this gain by accounting for aperture-induced frequency selectivity. These findings establish frequency selectivity as a key design consideration for large-scale PASS deployment.

\begin{appendix}
\subsection{Proof of Theorem \ref{thm:single_user_asymptotic_rate}}\label{proof_thm:single_user_asymptotic_rate}
Since ${\mathbf A}_N({\bm\psi})$ is Hermitian positive semidefinite, let its eigenvalues be denoted by $\lambda_1^{(N)},\lambda_2^{(N)},\ldots,\lambda_N^{(N)}\geq 0$. Using the fact that the determinant equals the product of the eigenvalues, we obtain $\det({\mathbf I}_N+{\mathbf A}_N({\bm\psi}))=\prod_{n=1}^{N}(1+\lambda_n^{(N)})$. Hence, $\frac{1}{N}I({\mathbf s}_N;{\mathbf r}_N)=\frac{1}{N}\sum_{n=1}^{N}\log_2(1+\lambda_n^{(N)})$. 

By construction, ${\mathbf G}_N({\bm\psi})$ is Toeplitz with entries $[{\mathbf G}_N({\bm\psi})]_{\hat n,n}=\hat g[\hat n-n;{\bm\psi}]$ for $\hat n,n\in[N]$. Its generating symbol is $\hat G({\rm e}^{{\rm j}\omega};{\bm\psi})\triangleq \sum_{\ell\in{\mathbb{Z}}}\hat g[\ell;{\bm\psi}]{\rm e}^{{\rm j}\omega\ell}$ for $\omega\in[-\pi,\pi]$. According to the properties of asymptotically equivalent Toeplitz matrices established by Gray \cite[Theorem 12]{gray2006toeplitz}, the product of Toeplitz matrices is asymptotically equivalent to a Toeplitz matrix generated by the product of their corresponding symbols. Since ${\mathbf G}_N^{\mathsf H}({\bm\psi})$ is generated by $\hat G^*({\rm e}^{{\rm j}\omega};{\bm\psi})$, the matrix sequence $\{{\mathbf A}_N({\bm\psi})\}$ is asymptotically equivalent to a Toeplitz matrix sequence with symbol $A({\rm e}^{{\rm j}\omega};{\bm\psi})=\gamma_0|\hat G({\rm e}^{{\rm j}\omega};{\bm\psi})|^2$.

Szeg{\H{o}}'s fundamental theorem \cite{grenander1958toeplitz} states that for a sequence of Hermitian Toeplitz matrices generated by a real-valued, essentially bounded symbol $A({\rm e}^{{\rm j}\omega};{\bm\psi})$, the arithmetic mean of a continuous function $F(\cdot)$ of its eigenvalues converges to the integral of the function evaluated on the symbol, namely,
\begin{align}\label{eq:Szego_final}
\lim_{N\rightarrow\infty}\frac{1}{N}\sum_{n=1}^{N}F(\lambda_n^{(N)})
=
\frac{1}{2\pi}\int_{-\pi}^{\pi}F(A({\rm e}^{{\rm j}\omega};{\bm\psi})){\rm d}\omega.
\end{align}
Applying \eqref{eq:Szego_final} with $F(x)=\log_2(1+x)$ to our normalized rate expression ${\mathsf R}({\bm\psi})=\lim_{N\rightarrow\infty}\frac{1}{N}I({\mathbf s}_N;{\mathbf r}_N)$, we obtain
\begin{align}\label{eq:normalized_MI_sum_eigs_final}
{\mathsf R}({\bm\psi})
=
\frac{1}{2\pi}\int_{-\pi}^{\pi}
\log_2(
1+\gamma_0|\hat G({\rm e}^{{\rm j}\omega};{\bm\psi})|^2)
{\rm d}\omega.
\end{align}

We next relate $\hat G({\rm e}^{{\rm j}\omega};{\bm\psi})$ to the continuous-time channel response. Since the discrete-time taps $\hat g[\ell;{\bm\psi}]$ are obtained by symbol-rate sampling of the continuous-time matched channel response $\hat g(t;{\bm\psi})=g(t;{\bm\psi})\ast c_p(t)$, the Poisson summation formula \cite{oppenheim1999discrete} gives
\begin{align}\label{poisson_summation_formula}
\hat G({\rm e}^{{\rm j}\omega};{\bm\psi})
=
\frac{1}{T_{\rm s}}\sum_{k=-\infty}^{\infty}
G_{\mathcal F}\left(\frac{\omega}{2\pi T_{\rm s}}-\frac{k}{T_{\rm s}};{\bm\psi}\right),
\end{align}
where $G_{\mathcal F}(f;{\bm\psi})\triangleq G(f;{\bm\psi})|P(f)|^2$ is the Fourier transform of $\hat g(t;{\bm\psi})$. Under the adopted signaling model, $T_{\rm s}=\frac{1}{B}$ and $|P(f)|^2=\frac{1}{B}\operatorname{rect}(\frac{f}{B})$. Thus, the shifted spectra in \eqref{poisson_summation_formula} do not overlap over the fundamental interval $\omega\in[-\pi,\pi]$, except at the band-edge points of measure zero. Consequently,
\begin{align}\label{eq:normalized_MI_sum_eigs_final_pre1}
\hat G({\rm e}^{{\rm j}\omega};{\bm\psi})
=
\frac{1}{T_{\rm s}}G_{\mathcal F}\left(\frac{\omega}{2\pi T_{\rm s}};{\bm\psi}\right)
=
B G(f;{\bm\psi})|P(f)|^2=G(f;{\bm\psi}),
\end{align}
where $f=\frac{\omega B}{2\pi}\in\left[-\frac{B}{2},\frac{B}{2}\right]$. By substituting \eqref{eq:normalized_MI_sum_eigs_final_pre1} into \eqref{eq:normalized_MI_sum_eigs_final} and using the change of variable $\omega=\frac{2\pi f}{B}$, the final result follows immediately. This completes the proof.
\subsection{Proof of Theorem \ref{them:rate_difference}}\label{proof_them:rate_difference}
Using $e^{-{\rm j}\omega\delta_{m}}=1-{\rm j}\omega\delta_{m}-\frac{\omega^2\delta_{m}^2}{2}+\frac{{\rm j}\omega^3\delta_{m}^3}{6}+\frac{\omega^4\delta_{m}^4}{24}+{\mathcal{O}}(\omega^5\delta_{m}^5)$, we obtain
\begin{align}\nonumber
\Phi_{\tau}(\omega)=1-\frac{\mu_2}{2}\omega^2+\frac{{\rm j}\mu_3}{6}\omega^3+\frac{\mu_4}{24}\omega^4+{\mathcal{O}}((\omega\Delta\tau)^5),
\end{align}
where $\Delta\tau=\tau_{\max}-\tau_{\min}$, and the linear term disappears because $\mu_1=0$. Taking the squared magnitude gives
\begin{align}\nonumber
\lvert\Phi_{\tau}(\omega)\rvert^2=1-\mu_2\omega^2+\left(\frac{\mu_2^2}{4}+\frac{\mu_4}{12}\right)\omega^4
+{\mathcal{O}}((\omega\Delta\tau)^6).
\end{align}
Therefore,
\begin{align}
\lvert G(f)\rvert^2&=A^2\left[1-\mu_2(2\pi f)^2+\left(\frac{\mu_2^2}{4}+\frac{\mu_4}{12}\right)(2\pi f)^4\right]+{\mathcal{O}}((f\Delta\tau)^6).\nonumber
\end{align}
It follows that
\begin{align}
1+\gamma_0\lvert G(f)\rvert^2&=(1+\rho)[1-a(2\pi f)^2+b(2\pi f)^4+{\mathcal{O}}((f\Delta\tau)^6)],\nonumber
\end{align}
where $a\triangleq\frac{\rho}{1+\rho}\mu_2$ and $b\triangleq\frac{\rho}{1+\rho}(\frac{\mu_2^2}{4}+\frac{\mu_4}{12})$. Using $\ln(1+x)=x-\frac{x^2}{2}+{\mathcal{O}}(x^3)$ with $x=-a(2\pi f)^2+b(2\pi f)^4+{\mathcal{O}}((f\Delta\tau)^6)$, we obtain
\begin{align}
\log_2(1+\gamma_0\lvert G(f)\rvert^2)=\log_2(1+\rho)-\frac{a}{\ln{2}}(2\pi f)^2
+\frac{1}{\ln{2}}\left(b-\frac{a^2}{2}\right)(2\pi f)^4+{\mathcal{O}}((f\Delta\tau)^6).\nonumber
\end{align}
For even $q$, $\frac{1}{B}\int_{-\frac{B}{2}}^{\frac{B}{2}}(2\pi f)^{q}{\rm d}f=\frac{(B\pi)^q}{q+1}$. We therefore obtain
\begin{align}
\Delta{\mathsf{R}}&={\mathsf{R}}_{\rm FF}-{\mathsf{R}}_{\rm FS}=\frac{1}{B}\int_{-\frac{B}{2}}^{\frac{B}{2}}
\log_2\left(
\frac{1+\rho}{1+\gamma_0|G(f)|^2}
\right){\rm d}f\nonumber\\
&=\frac{a\pi^2B^2}{3\ln{2}}-\frac{(b-\frac{a^2}{2})\pi^4B^4}{5\ln{2}}+{\mathcal{O}}((B\Delta\tau)^6)\nonumber\\
&=\frac{\pi^2B^2}{3\ln{2}}\frac{\rho\mu_2}{1+\rho}-\frac{\pi^4B^4}{5\ln{2}}\left[\frac{\rho\mu_4}{12(1+\rho)}+\frac{\rho(1-\rho)\mu_2^2}{4(1+\rho)^2}\right]+{\mathcal{O}}((B\Delta\tau)^6).\nonumber
\end{align}
Retaining the leading term and absorbing the fourth-order contribution into the remainder yields \eqref{rate_gap_second_order}.
This completes the proof.
\subsection{Asymptotic Rate under Joint Decoding}\label{proof_asym_joint_decoding}
The joint-decoding MI is given by $I({\mathbf s}_{1,N},\ldots,{\mathbf s}_{K,N};{\mathbf r}_N)
=
\log_2\det(
{\mathbf I}_N+\sum_{k=1}^{K}{\mathbf A}_{k,N}({\bm\psi})
)$, where ${\mathbf A}_{k,N}({\bm\psi})
=
\frac{P_k}{N_0}{\mathbf G}_{k,N}({\bm\psi}){\mathbf G}_{k,N}^{\mathsf H}({\bm\psi})$. For each user $k$, the matrix sequence $\{{\mathbf A}_{k,N}({\bm\psi})\}$ is asymptotically equivalent to a Toeplitz matrix sequence with generating symbol $A_k({\rm e}^{{\rm j}\omega};{\bm\psi})
=
\frac{P_k}{N_0}|\hat G_k({\rm e}^{{\rm j}\omega};{\bm\psi})|^2$ for $\omega\in[-\pi,\pi]$. Since the sum of Toeplitz matrices is again Toeplitz and asymptotic equivalence is preserved under addition \cite{gray2006toeplitz}, the matrix sequence $\{{\mathbf I}_N+\sum_{k=1}^{K}{\mathbf A}_{k,N}({\bm\psi})\}$ is asymptotically equivalent to a Toeplitz matrix sequence with symbol $1+\sum_{k=1}^{K}\frac{P_k}{N_0}|\hat G_k({\rm e}^{{\rm j}\omega};{\bm\psi})|^2$. Applying Szeg{\H{o}}'s theorem yields
\begin{align}\nonumber
{\mathsf R}_{\rm J}({\bm\psi})
=
\frac{1}{2\pi}\int_{-\pi}^{\pi}
\log_2\left(
1+\sum_{k=1}^{K}\frac{P_k|\hat G_k({\rm e}^{{\rm j}\omega};{\bm\psi})|^2}{N_0}
\right){\rm d}\omega.
\end{align}
Using the same Poisson-summation argument as in \eqref{poisson_summation_formula} and \eqref{eq:normalized_MI_sum_eigs_final_pre1}, we have $\hat G_k({\rm e}^{{\rm j}\omega};{\bm\psi})
=
G_k(f;{\bm\psi})$ for $f=\frac{\omega B}{2\pi}\in\left[-\frac{B}{2},\frac{B}{2}\right]$. Substituting this expression into ${\mathsf R}_{\rm J}({\bm\psi})$ and using the change of variable $\omega=\frac{2\pi f}{B}$, we obtain \eqref{joint_infinite_blocklength_rate_limit}. 

\subsection{Asymptotic Rate under Parallel Decoding}\label{proof_asym_parallel_decoding}
The achievable rate of user $k$ under parallel decoding is ${\mathsf R}_{k,{\rm P}}({\bm\psi})=\lim_{N\rightarrow\infty}\frac{1}{N}\log_2\det({\mathbf I}_N+{\mathbf A}_{k,N}({\bm\psi}){\mathbf Z}_{k,N}^{-1})$, where ${\mathbf A}_{k,N}({\bm\psi})=\frac{P_k}{N_0}{\mathbf G}_{k,N}({\bm\psi}){\mathbf G}_{k,N}^{\mathsf H}({\bm\psi})$ and ${\mathbf Z}_{k,N}= {\mathbf I}_N+\sum_{k'\neq k}{\mathbf{A}}_{k',N}({\bm\psi})$. For each user $k' \neq k$, the matrix sequence $\{{\mathbf A}_{k',N}({\bm\psi})\}$ is asymptotically equivalent to a Toeplitz matrix sequence with generating symbol $A_{k'}({\rm e}^{{\rm j}\omega};{\bm\psi})
=
\frac{P_{k'}}{N_0}|\hat G_{k'}({\rm e}^{{\rm j}\omega};{\bm\psi})|^2$ for $\omega\in[-\pi,\pi]$. Hence, $\{{\mathbf Z}_{k,N}\}$ is asymptotically equivalent to a Toeplitz matrix sequence with generating symbol $Z_k(\omega;{\bm\psi})
\triangleq
1+\sum_{k'\neq k}\frac{P_{k'}}{N_0}|\hat G_{k'}({\rm e}^{{\rm j}\omega};{\bm\psi})|^2$. Since $N_0>0$ and $P_{k'}\geq 0$, we have $Z_k(\omega;{\bm\psi})>0$ for all $\omega\in[-\pi,\pi]$. Therefore, by \cite[Theorem 11]{gray2006toeplitz}, the inverse sequence $\{{\mathbf Z}_{k,N}^{-1}\}$ is asymptotically equivalent to a Toeplitz matrix sequence generated by $\frac{1}{Z_k(\omega;{\bm\psi})}$.

Using again the product rule for asymptotically equivalent Toeplitz matrices \cite[Theorem 12]{gray2006toeplitz}, the matrix sequence $\{{\mathbf A}_{k,N}({\bm\psi}){\mathbf Z}_{k,N}^{-1}\}$ is asymptotically equivalent to a Toeplitz matrix sequence with generating symbol
\begin{align}\nonumber
\frac{A_k({\rm e}^{{\rm j}\omega};{\bm\psi})}{Z_k(\omega;{\bm\psi})}
=
\frac{|\hat G_k({\rm e}^{{\rm j}\omega};{\bm\psi})|^2P_k}
{N_0+\sum_{k'\neq k}|\hat G_{k'}({\rm e}^{{\rm j}\omega};{\bm\psi})|^2P_{k'}}.
\end{align}
Applying Szeg{\H{o}}'s theorem then yields
\begin{align}\label{eq:parallel_rate_omega_appendix}
{\mathsf R}_{k,{\rm P}}({\bm\psi})
=
\frac{1}{2\pi}\int_{-\pi}^{\pi}
\log_2\left(
1+\frac{A_k({\rm e}^{{\rm j}\omega};{\bm\psi})}{Z_k(\omega;{\bm\psi})}
\right){\rm d}\omega.
\end{align}
Using the same Poisson-summation argument as in \eqref{poisson_summation_formula} and \eqref{eq:normalized_MI_sum_eigs_final_pre1}, we have $\hat G_k({\rm e}^{{\rm j}\omega};{\bm\psi})
=
G_k(f;{\bm\psi})$ for $f=\frac{\omega B}{2\pi}\in\left[-\frac{B}{2},\frac{B}{2}\right]$. Substituting this expression into \eqref{eq:parallel_rate_omega_appendix} and using the change of variable $\omega=\frac{2\pi f}{B}$, we obtain \eqref{eq:parallel_rate_discrete}.
\end{appendix}
\bibliographystyle{IEEEtran}
\bibliography{mybib}

\begin{thebibliography}{10}
\providecommand{\url}[1]{#1}
\csname url@samestyle\endcsname
\providecommand{\newblock}{\relax}
\providecommand{\bibinfo}[2]{#2}
\providecommand{\BIBentrySTDinterwordspacing}{\spaceskip=0pt\relax}
\providecommand{\BIBentryALTinterwordstretchfactor}{4}
\providecommand{\BIBentryALTinterwordspacing}{\spaceskip=\fontdimen2\font plus
\BIBentryALTinterwordstretchfactor\fontdimen3\font minus
  \fontdimen4\font\relax}
\providecommand{\BIBforeignlanguage}[2]{{%
\expandafter\ifx\csname l@#1\endcsname\relax
\typeout{** WARNING: IEEEtran.bst: No hyphenation pattern has been}%
\typeout{** loaded for the language `#1'. Using the pattern for}%
\typeout{** the default language instead.}%
\else
\language=\csname l@#1\endcsname
\fi
#2}}
\providecommand{\BIBdecl}{\relax}
\BIBdecl

\bibitem{heath2025tri}
R.~W. Heath~Jr, J.~Carlson, N.~V. Deshpande, M.~R. Castellanos, M.~Akrout, and
  C.-B. Chae, ``The tri-hybrid {MIMO} architecture,'' \emph{IEEE Wireless
  Commun.}, vol.~33, no.~1, pp. 199--206, Feb. 2026.

\bibitem{zhu2024movable}
L.~Zhu, W.~Ma, and R.~Zhang, ``Movable antennas for wireless communication:
  Opportunities and challenges,'' \emph{IEEE Commun. Mag.}, vol.~62, no.~6, pp.
  114--120, Jun. 2024.

\bibitem{wong2020fluid}
K.-K. Wong, A.~Shojaeifard, K.-F. Tong, and Y.~Zhang, ``Fluid antenna
  systems,'' \emph{IEEE Trans. Wireless Commun.}, vol.~20, no.~3, pp.
  1950--1962, Mar. 2021.

\bibitem{shlezinger2021dynamic}
N.~Shlezinger, G.~C. Alexandropoulos, M.~F. Imani, Y.~C. Eldar, and D.~R.
  Smith, ``Dynamic metasurface antennas for {6G} extreme massive {MIMO}
  communications,'' \emph{IEEE Wireless Commun.}, vol.~28, no.~2, pp. 106--113,
  Apr. 2021.

\bibitem{wu2025intelligent}
Q.~Wu, G.~Chen, Q.~Peng, W.~Chen, Y.~Yuan, Z.~Cheng, J.~Dou, Z.~Zhao, and
  P.~Li, ``Intelligent reflecting surfaces for wireless networks: Deployment
  architectures, key solutions, and field trials,'' \emph{IEEE Wireless
  Commun.}, vol.~32, no.~6, pp. 141--148, Dec. 2025.

\bibitem{suzuki2022pinching}
A.~Fukuda, H.~Yamamoto, H.~Okazaki, Y.~Suzuki, and K.~Kawai, ``Pinching
  antenna: Using a dielectric waveguide as an antenna,'' \emph{NTT DOCOMO
  Technical J.}, vol.~23, no.~3, pp. 5--12, Jan. 2022.

\bibitem{liu2025pinching}
Y.~Liu, Z.~Wang, X.~Mu, C.~Ouyang, X.~Xu, and Z.~Ding, ``Pinching antenna
  systems {(PASS)}: Architecture designs, opportunities, and outlook,''
  \emph{IEEE Commun. Mag.}, vol.~64, no.~1, pp. 190--196, Jan. 2026.

\bibitem{liu2025pinchingtutorial}
Y.~Liu, H.~Jiang, X.~Xu, Z.~Wang, J.~Guo, C.~Ouyang, X.~Mu, Z.~Ding,
  A.~Nallanathan, G.~K. Karagiannidis \emph{et~al.}, ``Pinching-antenna systems
  ({PASS}): A tutorial,'' \emph{IEEE Trans. Commun.}, vol.~74, pp. 4881--4918,
  2026.

\bibitem{liu2025pinchingsurvey}
Y.~Liu, H.~Jiang, X.~Gan, X.~Xu, J.~Guo, Z.~Wang, C.~Ouyang, X.~Mu, Z.~Ding,
  A.~Nallanathan, O.~A. Dobre, G.~K. Karagiannidis, and R.~Schober, ``A survey
  of pinching-antenna systems ({PASS}),'' \emph{arXiv preprint
  arXiv:2601.18927}, 2026.

\bibitem{lu2026survey}
Y.~Lu, W.~Mao, X.~Xie, S.~Liao, Y.~Xu, Y.~Lin, R.~Zhang, B.~Ai, X.~Wang, and
  Z.~Ding, ``A survey on pinching-antenna systems: Optimization techniques,
  intelligent designs, and multifunctional applications,'' \emph{IEEE Commun.
  Surveys Tuts.}, vol.~28, pp. 6592--6627, 2026.

\bibitem{xu2026generalized}
Y.~Xu, J.~Cui, Y.~Zhu, Z.~Ding, T.-H. Chang, R.~Schober, V.~W. Wong, O.~A.
  Dobre, G.~K. Karagiannidis, H.~V. Poor \emph{et~al.}, ``Generalized
  pinching-antenna systems: A tutorial on principles, design strategies, and
  future directions,'' \emph{IEEE Commun. Surveys Tuts.}, vol.~28, pp.
  5872--5908, 2026.

\bibitem{beas2013millimeter}
J.~Beas, G.~Castanon, I.~Aldaya, A.~Arag{\'o}n-Zavala, and G.~Campuzano,
  ``Millimeter-wave frequency radio over fiber systems: a survey,'' \emph{IEEE
  Commun. Surveys Tuts.}, vol.~15, no.~4, pp. 1593--1619, 4th Quart. 2013.

\bibitem{wijewardhana2026wireless}
K.~R. Wijewardhana, A.~Yadav, M.~Zeng, M.~Elsayed, O.~A. Dobre, and Z.~Ding,
  ``Wireless-fed pinching-antenna systems ({Wi-PASS}) for {NextG} wireless
  networks,'' \emph{IEEE Commun. Mag.}, Early Access, 2026.

\bibitem{ouyang2025uplink}
C.~Ouyang, H.~Jiang, Z.~Wang, Y.~Liu, and Z.~Ding, ``Uplink and downlink
  communications in segmented waveguide-enabled pinching-antenna systems
  ({SWANs}),'' \emph{IEEE Trans. Commun.}, vol.~74, pp. 3688--3703, 2026.

\bibitem{wang2018spatial}
B.~Wang, F.~Gao, S.~Jin, H.~Lin, and G.~Y. Li, ``Spatial-and frequency-wideband
  effects in millimeter-wave massive {MIMO} systems,'' \emph{IEEE Trans. Signal
  Process.}, vol.~66, no.~13, pp. 3393--3406, Jul. 2018.

\bibitem{xiao2025frequency}
J.~Xiao, J.~Wang, M.~Zeng, Y.~Liu, and G.~K. Karagiannidis,
  ``Frequency-selective modeling and analysis for {OFDM}-integrated wideband
  pinching-antenna systems,'' \emph{IEEE Wireless Commun. Lett.}, vol.~14,
  no.~11, pp. 3500--3504, Nov. 2025.

\bibitem{wang2025modeling}
Z.~Wang, C.~Ouyang, X.~Mu, Y.~Liu, and Z.~Ding, ``Modeling and beamforming
  optimization for pinching-antenna systems,'' \emph{IEEE Trans. Commun.},
  vol.~73, no.~12, pp. 13\,904--13\,919, Dec. 2025.

\bibitem{gallager2008principles}
R.~G. Gallager, \emph{Principles of Digital Communication}.\hskip 1em plus
  0.5em minus 0.4em\relax Cambridge, U.K.: Cambridge Univ. Press, 2008.

\bibitem{gazzah2001asymptotic}
H.~Gazzah, P.~A. Regalia, and J.-P. Delmas, ``Asymptotic eigenvalue
  distribution of block {Toeplitz} matrices and application to blind {SIMO}
  channel identification,'' \emph{IEEE Trans. Inf. Theory}, vol.~47, no.~3, pp.
  1243--1251, Mar. 2001.

\bibitem{ouyang2024primer}
C.~Ouyang, Z.~Wang, Y.~Chen, X.~Mu, and P.~Zhu, ``A primer on near-field
  communications for next-generation multiple access,'' \emph{Proc. {IEEE}},
  vol. 112, no.~9, pp. 1527--1565, Sep. 2024.

\bibitem{yeh2008essence}
C.~Yeh and F.~I. Shimabukuro, \emph{The Essence of Dielectric
  Waveguides}.\hskip 1em plus 0.5em minus 0.4em\relax New York, NY, USA:
  Springer, 2008.

\bibitem{ding2024flexible}
Z.~Ding, R.~Schober, and H.~V. Poor, ``Flexible-antenna systems: A
  pinching-antenna perspective,'' \emph{IEEE Trans. Commun.}, vol.~73, no.~10,
  pp. 9236--9253, Oct. 2025.

\bibitem{ouyang2025array}
C.~Ouyang, Z.~Wang, Y.~Liu, and Z.~Ding, ``Array gain for pinching-antenna
  systems ({PASS}),'' \emph{IEEE Commun. Lett.}, vol.~29, no.~6, pp.
  1471--1475, Jun. 2025.

\bibitem{xu2024rate}
Y.~Xu, Z.~Ding, and G.~K. Karagiannidis, ``Rate maximization for downlink
  pinching-antenna systems,'' \emph{IEEE Wireless Commun. Lett.}, vol.~14,
  no.~5, pp. 1431--1435, May 2025.

\bibitem{gray2003asymptotic}
R.~M. Gray, ``On the asymptotic eigenvalue distribution of {Toeplitz}
  matrices,'' \emph{IEEE Trans. Inf. Theory}, vol.~18, no.~6, pp. 725--730,
  Nov. 1972.

\bibitem{gray2006toeplitz}
------, ``Toeplitz and circulant matrices: A review,'' \emph{Found. Trends
  Commun. Inf. Theory}, vol.~2, no.~3, pp. 155--239, 2006.

\bibitem{xu2026pinching}
Y.~Xu, Z.~Ding, R.~Schober, and T.-H. Chang, ``Pinching-antenna systems with
  in-waveguide attenuation: Performance analysis and algorithm design,''
  \emph{IEEE Trans. Wireless Commun.}, vol.~25, pp. 14\,564--14\,580, 2026.

\bibitem{steele2004cauchy}
J.~M. Steele, \emph{The Cauchy-Schwarz Master Class: An Introduction to the Art
  of Mathematical Inequalities}.\hskip 1em plus 0.5em minus 0.4em\relax New
  York, NY, USA: Cambridge Univ. Press, 2004.

\bibitem{verdu1989multiple}
S.~Verdu, ``Multiple-access channels with memory with and without frame
  synchronism,'' \emph{IEEE Trans. Inf. Theory}, vol.~35, no.~3, pp. 605--619,
  May 1989.

\bibitem{heath2018foundations}
R.~W. Heath~Jr and A.~Lozano, \emph{Foundations of {MIMO} Communication}.\hskip
  1em plus 0.5em minus 0.4em\relax Cambridge, U.K.: Cambridge Univ. Press,
  2018.

\bibitem{hassibi2003much}
B.~Hassibi and B.~M. Hochwald, ``How much training is needed in
  multiple-antenna wireless links?'' \emph{IEEE Trans. Inf. Theory}, vol.~49,
  no.~4, pp. 951--963, Apr. 2003.

\bibitem{falconer2002frequency}
D.~Falconer, S.~L. Ariyavisitakul, A.~Benyamin-Seeyar, and B.~Eidson,
  ``Frequency domain equalization for single-carrier broadband wireless
  systems,'' \emph{IEEE Commun. Mag.}, vol.~40, no.~4, pp. 58--66, Apr. 2002.

\bibitem{pancaldi2008single}
F.~Pancaldi, G.~M. Vitetta, R.~Kalbasi, N.~Al-Dhahir, M.~Uysal, and H.~Mheidat,
  ``Single-carrier frequency domain equalization,'' \emph{IEEE Signal Process.
  Mag.}, vol.~25, no.~5, pp. 37--56, Sep. 2008.

\bibitem{polyanskiy2010channel}
Y.~Polyanskiy, H.~V. Poor, and S.~Verd{\'u}, ``Channel coding rate in the
  finite blocklength regime,'' \emph{IEEE Trans. Inf. Theory}, vol.~56, no.~5,
  pp. 2307--2359, May 2010.

\bibitem{grenander1958toeplitz}
U.~Grenander and G.~Szeg{\H{o}}, \emph{Toeplitz Forms and Their
  Applications}.\hskip 1em plus 0.5em minus 0.4em\relax Oakland, CA, USA: Univ.
  of California Press, 1958.

\bibitem{oppenheim1999discrete}
A.~V. Oppenheim and R.~W. Schafer, \emph{Discrete-Time Signal Processing},
  3rd~ed.\hskip 1em plus 0.5em minus 0.4em\relax Upper Saddle River, NJ, USA:
  Pearson, 2010.

\end{thebibliography}
\end{document}